\documentclass[runningheads]{llncs}
\usepackage{esvect}
\usepackage[T1]{fontenc}
\usepackage{newtxtext}       % 
\usepackage[varvw]{newtxmath}       % selects Times Roman as basic font

\usepackage{graphicx}
\usepackage{braket}
\usepackage{mathtools}
\usepackage{mathpartir}
\usepackage{bussproofs}
\usepackage{makecell}
\usepackage{stmaryrd}
\usepackage{multirow}
\usepackage{tcolorbox} % needed for colored boxes
\tcbuselibrary{skins,breakable} % optional, for nicer boxes
\usepackage{pdflscape}
\usepackage{adjustbox} % <— for automatic fitting
\usepackage{listings}
\usepackage{xcolor}
\usepackage{subcaption}    % for subfigures
\usepackage{hyperref}
\usepackage{algorithm}
\usepackage{algpseudocode}
\usepackage{float}
\usepackage{color,soul}

\lstdefinestyle{pseudocode}{
    basicstyle=\ttfamily\scriptsize,          % monospaced font
    keywordstyle=\color{blue!70!black},  % keywords in blue
    commentstyle=\color{gray!70!black},  % comments in gray
    stringstyle=\color{orange!90!black}, % (if you use strings)
    columns=fullflexible,                % no weird spacing
    keepspaces=true,                     % preserve indentation
    showstringspaces=false,              % don't show spaces as dots
    frame=none,                          % no box around code
    xleftmargin=1em,                     % small left indent
    escapeinside={(*@}{@*)},             % allow LaTeX inside
    aboveskip=0.5em,
    belowskip=0.5em,
    breaklines=true,  
    backgroundcolor=\color{white},
    frame=single,
}

\newcommand{\Reals}[0]{\mathbb{R}}
\newcommand{\allops}[1]{\mathcal{L}(#1)}
\newcommand{\bstateS}[0]{\Phi}
\newcommand{\bstateSX}[0]{\Psi}
\newcommand{\bstate}[0]{\bstateS^{+}}
\newcommand{\bstateZ}[0]{\bstateS^{-}}

\newcommand{\nQHL}[0]{\textnormal{\textsc{nQHL}}}

\newcommand{\weakp}[3]{\texttt{wp}^{#1}\left(#2, #3\right)}
\newcommand{\mvals}[1]{\vec{#1}}

\newcommand{\probOf}[1]{\mathbb{P}\left(#1\right)}
\newcommand{\condProb}[2]{\probOf{#2 \mid #1}}
\newcommand{\topure}[2]{#1(#2)}

\newcommand{\ensembleAt}[2]{#1[{#2}]}
\newcommand{\branchAt}[3]{#1\downarrow^{#2}_{#3}}
\newcommand{\csubs}[3]{#1[#2 \leftarrow #3]}
\newcommand{\qsubs}[3]{\branchAt{#1}{}{#3}}
\newcommand{\subsfunc}[3]{\texttt{R}(#3,{#2}^{#1})}

\newcommand{\rcond}[2]{\branchAt{#1}{}{#2}}
\newcommand{\norm}[1]{\lVert #1 \rVert}
\newcommand{\conv}[0]{\operatorname{conv}}

\newcommand{\real}[0]{r}
\newcommand{\horizon}[0]{k}
\newcommand{\precision}[0]{\epsilon}
\newcommand{\ensemble}[0]{\beta}
\newcommand{\prob}[0]{p}
\newcommand{\op}[0]{M}
\newcommand{\qstate}[0]{\alpha}
\newcommand{\bit}[0]{b}
\newcommand{\cstate}[0]{c}
\newcommand{\hstate}[0]{h}
\newcommand{\cstateSize}[0]{m}
\newcommand{\hsdim}[0]{n}
\newcommand{\numQubits}[0]{N}

\newcommand{\projector}[0]{\Lambda}
\newcommand{\hardware}[0]{H}

\newcommand{\Id}[0]{\mathbb{I}}

\newcommand{\X}[0]{\texttt{X}}
\newcommand{\Z}[0]{\texttt{Z}}
\newcommand{\Had}[0]{\texttt{H}}
\newcommand{\CX}[0]{\texttt{CX}}

\newcommand\sem[1]{\llbracket #1\rrbracket}
\newcommand{\halt}[0]{\texttt{halt}} 
\newcommand{\Skip}[0]{\texttt{skip}}
\newcommand{\Measure}[0]{\texttt{measure}}
\newcommand{\IfElse}[3]{\texttt{if}~(#1)~\{#2\}~\texttt{else}~\{#3\}}
\newcommand{\MyWhile}[2]{\texttt{while } (#1)~\{ #2 \}}

\newcommand{\BTerm}[0]{\texttt{Boolean term}}
\newcommand{\BExpr}[0]{\texttt{Boolean expression}}
\newcommand{\QTerm}[0]{\texttt{Quantum term}}
\newcommand{\PTerm}[0]{\texttt{Probabilistic term}}
\newcommand{\StatesA}[0]{\texttt{State assertion}}
\newcommand{\EnsA}[0]{\texttt{Ensemble assertion}}

\newcommand{\bexpr}[0]{\chi}
\newcommand{\bterm}[0]{\tau}
\newcommand{\qterm}[0]{\varrho}
\newcommand{\pterm}[0]{\lambda}
\newcommand{\stateA}[0]{\varphi}
\newcommand{\ensA}[0]{\post}

\newcommand{\ensembles}[0]{\mathcal{D}(\hstates)}
\newcommand{\allensZ}[1]{\mathcal{D}_{\bot}(#1)}

\newcommand{\zens}[0]{\mathbb{O}}

\newcommand{\hs}[0]{\mathcal{H}}
\newcommand{\allqv}{Q}
\newcommand{\allcv}{\mathcal{X}}
\newcommand{\cstates}[0]{\mathcal{C}}
\newcommand{\hstates}[0]{\mathcal{M}_{\allqv,\allcv}}

\newcommand{\qvar}[0]{q}
\newcommand{\qvars}[0]{\mvals{\qvar}}
\newcommand{\cvar}[0]{x}
\newcommand{\cvars}[0]{\mvals{\cvar}}

\newcommand{\instruction}[0]{I} % a single instruction
\newcommand{\allInstructions}[0]{\mathcal{I}}

\newcommand{\schannel}[0]{\zeta}
\newcommand{\channel}[1]{\schannel^{#1}_\hardware} % an specific quantum channel, pass as parameter an instruction

\newcommand{\pre}[0]{\phi}
\newcommand{\post}[0]{\psi}
\newcommand{\prog}[0]{P} 
\newcommand{\pprog}[0]{\widetilde{\prog}}
\newcommand{\pprogs}[0]{\progs^{\oplus}}
\newcommand{\cprog}[0]{\overline{\prog}} % candidate program
\newcommand{\atomprog}[0]{P^{\bullet}} 

\newcommand{\seqobs}[1]{\Theta(#1)}
\newcommand{\guard}[0]{G}

\newcommand{\pomdpTransF}[0]{\delta}
\newcommand{\belief}[0]{\widetilde{\ensemble}}

\newcommand{\allBeliefs}[0]{B}

\newcommand{\beliefGraph}[0]{\mathcal{B}}
\newcommand{\supp}[0]{\operatorname{Supp}}

\newcommand{\IPMAIns}[0]{\text{IPM}}
\newcommand{\IPMAInsP}[0]{\text{IPM}}

\newcommand{\IPMAPIns}[0]{\text{IPM}}
\newcommand{\Paritybitflip}[0]{{\sc Parity-bit-flip}}
\newcommand{\PhasePrep}[0]{{\sc Local phase-preparation}}
\newcommand{\IPMAPAbbr}[0]{\sc{Parity-bit.}}
\newcommand{\CXHIns}[0]{\text{CX+H}}

\newcommand{\StateReset}[0]{{\sc State-reset}}
\newcommand{\ResetIns}[0]{\text{Reset}}
\newcommand{\ResetAbbr}[0]{{\sc State-reset}}
\newcommand{\GHZsp}[0]{{\sc GHZ~state-preparation}}
\newcommand{\GHZIns}[0]{\text{GHZ}}
\newcommand{\GHZAbbr}[0]{\makecell[c]{\sc{GHZ}\\ \sc{state-prep.}}}
\newcommand{\LPhaseIns}[0]{\text{LPhase}}
\newcommand{\LPhaseAbbr}[0]{\makecell[c]{\sc{Loc. phase-prep.}}}

\newcommand{\htriple}[4]{\left\{#1\right\}~#2^{#4}~\left\{#3\right\}}

\newcommand{\MC}[0]{M}
\newcommand{\QMCTransitionF}[0]{\Delta}
\newcommand{\polytope}[0]{\mathcal{P}}
\newcommand{\QSetPolytopes}[0]{\Omega}
\newcommand{\progs}[0]{\mathcal{A}}
\newcommand{\atomprogs}[0]{\mathcal{A}^{\bullet}}
\newcommand{\finprogs}[0]{\overline{\progs}}
\newcommand{\cprogs}[0]{\finprogs}
\newcommand{\fincprogs}[0]{\cprogs_{\horizon}}
\newcommand{\config}[0]{\kappa}
\newcommand{\configs}[0]{\mathcal{K}}
\newcommand{\synthprogs}[0]{\mathbb{J}}
\newcommand{\synthins}[0]{J}

\newcommand{\plength}[2]{|#1|_{#2}}
\newcommand{\cstatesOf}[1]{\seqobs{#1}}

\newcommand{\prepolys}[0]{\operatorname{Poly}}
\newcommand{\getFinalEnsemble}[0]{\textnormal{\sc finalEns}}
\newcommand{\ReachConfs}[0]{\textnormal{\sc ReachConfs}}
\newcommand{\DCandidates}[0]{\operatorname{DCP}}
\newcommand{\CBlock}[0]{\operatorname{CBlock}}
\newcommand{\nextI}[0]{\operatorname{NextI}}

\newcommand{\initSet}[0]{S}
\newcommand{\targetSet}[0]{T}

\begin{document}
\title{Noise-aware Verification and Synthesis of Quantum Programs}
\author{Stefanie Muroya\inst{1} \and
Krishnendu Chatterjee\inst{1} \and
Thomas A. Henzinger\inst{1}}
% Follow the naming convention in which the surname is the last name.
% In this field, provide the full first names (not only the initials).
% Do not include academic titles (e.g., Prof. or Dr.).
% Springer encourages the inclusion of author ORCIDs.
%
\authorrunning{Muroya et al.}
\institute{Institute of Science and Technology in Austria (ISTA), Klosterneuburg, Austria}
\maketitle              % typeset the header of the contribution
\begin{abstract}
This work provides a comprehensive framework for noise-aware quantum program analysis, ranging from logical foundations to model checking and automated synthesis.
While most research on quantum programming considers an idealized, noise-free semantics for quantum programs, we reason about quantum programs that are executed on real, noisy hardware.
We use the error models published by quantum-hardware vendors to give a hardware-dependent semantics to quantum programs.
We then present a noise-aware quantum Hoare logic with a relatively complete proof system. 
We finally use state-space exploration techniques for the hardware-aware verification and synthesis of loop-free quantum programs. 
In this way, we synthesize optimal, hardware-dependent solutions for parity bit-flip, GHZ state preparation, and other problems that commonly occur as subroutines in larger quantum algorithms.
We evaluate our method on all 55 hardware specifications provided by the IBM Qiskit toolkit.
\end{abstract}

\section{Introduction}
Quantum computing is fundamentally constrained by the error rates of quantum devices.
Each device can be described by a mathematical error model, called a \emph{hardware specification}, which characterizes the probability of errors associated with each quantum instruction.
Hardware vendors~\cite{Qiskit,Cirq,Quantinuum} publish these specifications based on calibration experiments.
Although today's quantum computers remain noisy~\cite{Preskill18}, most existing logics~\cite{Feng07,Ying11,ying24,chadha06,Wu2025} for quantum programs assume idealized, noise-free semantics, limiting their applicability to real hardware.

We are interested in \emph{noise-aware} and \emph{modular} quantum program analysis.
Existing compositional approaches often obtain modularity through separability-based principles~\cite{Zhou21}. Semantics based solely on density matrices~\cite{hung19,deng22,Ying11,unruh19qrhl}  discard ensemble-realization information that can be relevant for modular reasoning.
This work presents a formal framework that achieves modular noise-aware reasoning without separability assumptions using  \emph{ensembles}~\cite{qc_bible} as semantic primitive---i.e., probability distributions with finite support over hybrid classical-quantum states---instead of density matrices.

We argue that ensembles offer the required expressivity for noise-aware modular verification, for two reasons:
(1)~while every ensemble determines a unique density matrix, a density matrix may admit many ensemble realizations, but the concrete ensemble realization of a local density matrix may be relevant for the global system; and
(2)~hardware specifications published by vendors are formulated in terms of probabilistic instruction errors, which induce a canonical ensemble realization under the chosen error model (whereas density-matrix semantics discard this realization).
The following example illustrates these two reasons.

\begin{example}[Local phase preparation] 
We use Dirac's bra-ket notation, and write $\ket{\cdot}$ for a complex column vector,
and $\bra{\cdot}$ for the conjugate transpose. Consider the qubit computational basis states $\ket{0} = \begin{bmatrix}
    1 & 0
\end{bmatrix}^T$ and $\ket{1} = \begin{bmatrix}
    0 & 1
\end{bmatrix}^T$, the two Bell states $\ket{\bstateS^{\pm}} =(\ket{00}\pm\ket{11})/\sqrt{2}$ on two qubits, and denote by $\Id$ the identity matrix.
While the local density matrices of either subsystem are the same for both Bell states, the  representation of the global state induces different local ensemble  realizations:
\[
\begin{tabular}{c|c|c}
    global state &  local density matrix & local ensemble \\
    \hline\hline
    $\ket{\bstate}$ & $\Id/2$ & $\{\ket{0}: 0.5, \ket{1}: 0.5\}$\\
    \hline
    $\ket{\bstateZ}$ & $\Id/2$ & $\{\ket{0}: 0.5, -\ket{1}: 0.5\}$\\
\end{tabular}
\]
showing that different ensemble realizations of the same local density matrix are relevant and visible on the global state. 
Moreover, hardware specifications induce a particular ensemble realization.
For example, consider a subprocedure that applies the phase-flip instruction $\instruction = \Z(\qvar)$, and two hardware specifications:
$\hardware_0(\instruction) = \{\sqrt{0.1}\,\X_\qvar,\ \sqrt{0.9}\,\Id\}$, which applies a bit-flip with probability $0.1$ after executing $\instruction$,
and the noiseless specification $\hardware_1(\instruction) = \{\Id\}$.
We can prove that the Hoare triple
\begin{equation}
\htriple{\probOf{\Id = \ket{0}} = 0.5 \land \probOf{\Id = \ket{1}} = 0.5}{\Z(\qvar)}{\probOf{\Id = -\ket{1}} \geq 0.5}{\hardware_1}
\label{eq:mot_example}
\end{equation}
is valid for $\hardware_1$ but invalid for $\hardware_0$.
Local phase preparation is a standard  subprocedure in protocols~\cite{Bennett93,superdense92}, algorithms~\cite{qc_bible}, and quantum error-correcting codes~\cite{gottesman97}.
Yet its verification is inexpressible with density matrices, since the reduced density matrix is $\Id/2$ both before and after executing $\instruction$, under either specification. 
\label{ex:motivation}
\end{example}

\paragraph{Contributions.}
We make four main contributions.
{\em First}, we establish logical foundations for the noise-aware analysis of quantum programs in Section~\ref{sec:nqhl}, where we introduce the {\em noise-aware quantum Hoare logic} nQHL, together with a proof system that is relatively complete for partial correctness.
We use an assertion language for quantum programs which is interpreted over ensembles.
Three special forms of ensemble assertions are noteworthy:
a {\em singular assertion} specifies a single ensemble; 
a {\em linear assertion} defines a finite union of polytopes of ensembles;
a {\em target assertion} specifies a lower bound on the probability of a ``target set'' of  classical-quantum states.
{\em Second}, Section~\ref{sec:verification} presents a numerically precise algorithm for the automated verification of a noise-aware Hoare triple that consists of a loop-free quantum program, a linear precondition, and an arbitrary postcondition.
This bounded model checking is performed on an infinite Markov chain over ensembles.
{\em Third}, Section~\ref{sec:synthesis} presents an algorithm to find a loop-free program with at most $\horizon$ instructions which, when executed from any ensemble that satisfies the precondition, yields an ensemble that satisfies the postcondition.
For linear preconditions, we reduce this bounded synthesis to 
dominant-program enumeration in order to avoid the brute-force enumeration of all programs of length at most~$\horizon$.
If the precondition specifies a unique initial ensemble and the postcondition specifies a target set, then an optimal bounded-reachability policy can be computed on an ensemble Markov Decision Process (MDP).
This policy can subsequently be translated into a {\em hardware-optimal quantum program}, which is a program that---among all loop-free programs with at most $\horizon$ instructions---reaches the target set on the given noisy hardware from the given initial ensemble with the highest probability~\cite{Muroya25}.
The guarantee of hardware-optimality makes this method valuable for quantum compilers.
{\em Fourth}, in Section~\ref{sec:experiments} we report the experimental results of running our bounded-synthesis tool for four common tasks from quantum computing on 55 IBM hardware specifications accessed through Qiskit~\cite{Qiskit} v2.2.1.
While \cite{Muroya25} considers only the special case of a single initial ensemble and a set of target states,
our tool handles arbitrary linear preconditions and arbitrary postconditions.
This allows us to synthesize hardware-optimal quantum subroutines  
for \Paritybitflip,~\StateReset, and \GHZsp~formulations from~\cite{Muroya25},
but also for the \PhasePrep~problem, 
and generalizations of these tasks to multiple initial ensembles.
These problems are common subroutines in larger quantum algorithms~\cite{Hillery98,Geher24,Bennett93,superdense92,gottesman97}. Our results demonstrate that optimal program synthesis depends on the hardware specification and can yield solutions that differ substantially from textbook implementations designed for noise-free hardware. Thus, while our method targets short, loop-free quantum programs, its compositional nature makes it relevant in practice.
We further validate our results through numerical hardware simulations.

\paragraph{Related work.}
We extend the \textsf{QIMDP} programming language~\cite{deng22} with ensemble semantics that account for hardware specifications. Unlike prior noise-aware verification work~\cite{hung19,Tao21}, we also consider measurement errors and use ensembles instead of density matrices.
Most works on the formal verification~\cite{chen23,Nengkun21,amy25} and synthesis~\cite{Kang23,Sohaib23,Xiao21,Furrutter24,Deng24,Xu23,ZiHao24} of quantum programs target idealized, noise-free settings.
Existing Hoare logics for quantum programs are largely based on density-matrix semantics~\cite{hondt06,Ying11,Feng07,zhou19,liu19,ying24,Wu2025,kakutani09,sun24Q,deng22}. Chadha et al.~\cite{chadha06} consider an ensemble semantics, but their assertion language can only reason about measurement statistics and cannot constrain  amplitudes of quantum states. Quantum separation logic~\cite{Zhou21} achieves modular verification through a separating conjunction interpreted over separable quantum states. To the best of our knowledge, nQHL is the first quantum Hoare-style logic whose assertions are interpreted over arbitrary finite-support ensembles and can reason simultaneously about concrete ensemble realizations and standard measurement statistics, without separability assumptions.

  % why we need new semantics
\section{Framework}
\label{sec:framework}
\paragraph{Hilbert spaces.} We consider quantum systems of $\numQubits$ qubits, with complex Hilbert space $\hs$ with vectors of dimension~$\hsdim=2^{\numQubits}$.
The inner product of two vectors $\ket{\qstate}$ and $\ket{\qstate'}$ is denoted  by $\braket{\qstate\mid\qstate'}$ and its norm by $\norm{\qstate} = \sqrt{\braket{\qstate\mid\qstate}}$.
If  $\ket{\qstate_0}$ and $\ket{\qstate_1}$ are two vectors residing in Hilbert spaces $\hs_0$ and $\hs_1$ of dimensions $\hsdim_0$ and $\hsdim_1$, respectively,
their tensor product $\ket{\qstate_0} \otimes \ket{\qstate_1}$,
which we abbreviate by $\ket{\qstate_0\qstate_1}$, is a vector of dimension~$\hsdim_0\cdot \hsdim_1$ that resides in $\hs_0\otimes\hs_1$.
An {\em operator} $\op$ acting on $\hs$ is a complex square matrix of dimension $\hsdim\times\hsdim$. 
One can construct a linear operator from a vector $\ket{\qstate} \in \hs$ via the outer product denoted by $\ket{\qstate}\bra{\qstate}$. 
The outer product $\ket{\qstate}\bra{\qstate}$ is defined on all vectors $\ket{\qstate'} \in \hs$ as $ (\ket{\qstate}\bra{\qstate})\ket{\qstate'} = \braket{\qstate\mid \qstate'}\ket{\qstate}$.

\paragraph{Quantum states.}  
A closed quantum system that does not interact with its environment can be 
fully described by a {\em pure quantum state}, i.e., a unit vector $\ket{\qstate}$  in~$\hs$ with $\norm{\qstate} = 1$.
In particular, a qubit state is a unit vector of dimension 2.
A relevant orthonormal basis for qubit systems is the {\em computational basis} that consists of the two unit vectors $\ket{0}$ and $\ket{1}$ defined in  Example~\ref{ex:motivation}.
The state space of a multi-qubit system is the tensor product of the single-qubit spaces.

\subsection{Hybrid quantum-classical memory}
Given a finite and totally ordered set $\allqv$ of $\numQubits$ quantum variables,
the subsystem associated with a quantum variable $\qvar \in \allqv$ with local Hilbert space of dimension 2 is denoted by $\hs_{\qvar}$, and
the Hilbert space of the combined system is $\hs_{\allqv} = \bigotimes_{\qvar \in \allqv}\hs_\qvar$.
The set of all linear operators  acting on $\hs_{\allqv}$ is denoted by $\allops{\allqv}$.

We also consider classical systems of $\cstateSize$ classical bits,
for a set $\allcv$ of $\cstateSize$ classical variables. 
A {\em classical state} is a function from $\allcv$ to $\{0,1\}$.
Let $\cstates_\allcv$ be the set of all $2^\cstateSize$ classical states.
A {\em hybrid state} $(\ket{\qstate},\cstate)$ consists of a quantum state $\ket{\qstate}\in\hs_{\allqv}$ and a classical state $\cstate\in\cstates_\allcv$. 
We denote by
$\hstates = \hs_{\allqv} \times \cstates_{\allcv}$ the set of all hybrid states.

In general, ensembles (or density matrices) are required to capture the probabilistic uncertainty that results from quantum errors.
An {\em ensemble} $\ensemble=\{\hstate_i: \prob_i  \}$, with $\sum_i \prob_i=1$ and $\prob_i \geq 0$, is a distribution over hybrid states with finite support,
where the support, denoted by $\supp(\ensemble) \subset \hstates$, is the set of hybrid states that have positive probability in~$\ensemble$.
If $\sum_i \prob_i \leq 1$, then $\ensemble$ is called a {\em subensemble}.
The set of all ensembles over hybrid states is denoted by 
$\ensembles$, and $\allensZ{\hstates} = \ensembles \cup \{\zens\}$ denotes the lifted set that includes the zero subensemble $\zens$, i.e., $\zens(\hstate) = 0$ for all $\hstate\in\hstates$.

\paragraph{Notation.} Throughout this work, $\qvars$ denotes a tuple of distinct quantum variables and $\mvals{\allqv}$ the tuple that consists of all quantum variables in $\allqv$ in a canonical order.
We adopt the convention from the quantum computing literature when operators apply to part of the quantum system~\cite{qc_bible}:
given an operator $\op$ acting on $\hs_{\qvar}$ we denote by $\op_{\qvar}$ the operator that applies $\op$ to $\qvar$ and the identity to the remaining variables. The notation is similar for tuples of quantum variables.
For a hybrid state $\hstate = (\ket{\qstate}, \cstate)$, we denote the classical state and quantum state of $\hstate$ by $\topure{\hstate}{\allcv} = \cstate$ and  $\topure{\hstate}{\allqv} = \ket{\qstate}$, respectively.
For a classical variable $\cvar \in \allcv$, we write $\hstate(\cvar) = \cstate(\cvar)$ to denote its value in $\hstate$.
For ensembles $\ensemble, \ensemble'\in \ensembles$ over hybrid states, addition and multiplication with a scalar $ \real\in \Reals$ are element-wise:
$(\ensemble + \ensemble')(\hstate) = \ensemble(\hstate) + \ensemble'(\hstate)$ and 
$(\real\cdot \ensemble)(\hstate) = \real\cdot \ensemble(\hstate)$ for all $\hstate \in \hstates$.

\subsection{Atomic instructions and noise-free semantics}
We have three types of atomic instructions: (i) classical instructions, (ii) unitary instructions, which operate on quantum states, and (iii) measurement instructions, which involve both quantum and classical states.
Each instruction maps a hybrid state to an ensemble. Consider a hybrid state $\hstate \in \hstates$.

A {\em classical instruction} is a write operation defined as a pair $\langle \cvar, \bit\rangle$ consisting of a classical variable $\cvar \in \allcv$ and a bit $\bit \in \{0,1\}$.
The semantics gives a {\em deterministic} ensemble $\sem{\langle\cvar, \bit\rangle}(\hstate) = \{(\topure{\hstate}{\allqv}, \cstate'): 1\}$ that assigns unit probability to the hybrid state with the same quantum state as $\hstate$ and the classical state $\cstate'$ defined by  $\cstate'(\cvar') = \bit$ if $\cvar' = \cvar$, otherwise $\cstate'(\cvar') = \cstate(\cvar')$ for all classical variables $\cvar' \in \allcv$.

Given an operator $\op$ on $\hs_{\allqv}$, denote the evolution of the quantum state of $\hstate$ under $\op$ as   $\branchAt{\hstate}{}{\op} = \left(\frac{\op\cdot\topure{\hstate}{\allqv}}{\norm{\op\cdot \topure{\hstate}{\allqv}}}, \topure{\hstate}{\allcv}\right)$ if $\norm{\op\cdot \topure{\hstate}{\allqv}} \neq 0$. Otherwise $\branchAt{\hstate}{}{\op}=\left(\mvals{0}, \topure{\hstate}{\allcv}\right)$, where $\mvals{0}$ is the zero vector.
The evolution of an ideal noise-free quantum system can be described by a unitary matrix $U$ satisfying $U^\dagger U = UU^\dagger = \Id$ which preserves the norm of vectors~\cite{qc_bible}. 
Accordingly, a {\em unitary instruction} is a pair $\langle U, \qvars \rangle$ consisting of a unitary matrix $U$ and a tuple $\qvars$ of distinct quantum variables, such that $U$ is an operator on  $\hs_{\qvars}$.
A unitary instruction executed on a hybrid state gives the unit ensemble  $\sem{\langle U, \qvars \rangle}(\hstate) = \{\branchAt{\hstate}{}{U_{\qvars}}: 1\}$.

We define a {\em measurement instruction} as a pair $\langle \qvar, \cvar \rangle$ consisting of 
a quantum variable $\qvar\in\allqv$ and a classical variable $\cvar \in \allcv$, 
which indicate the qubit we wish to measure and the classical bit that stores the result.
Intuitively, the semantics of measurement instructions correspond to single-qubit measurements in the computational basis, as supported by current hardware, which are described in terms of two orthogonal projectors, $\projector_0 = (\ket{0}\bra{0})_{\qvar}$ and $\projector_1 = (\ket{1}\bra{1})_{\qvar}$.
The noiseless execution of the measurement instruction is defined as $\sem{\langle \qvar, \cvar \rangle}(\hstate) = \sum_{\bit \in \{0,1\}}\norm{\projector_\bit\cdot \topure{\hstate}{\allqv}}^2 \cdot \sem{\langle\cvar, \bit \rangle}(\branchAt{\hstate}{}{\projector_\bit})$.
Intuitively, $\norm{\projector_\bit\cdot \topure{\hstate}{\allqv}}^2$ is the probability of receiving $\bit$ as measurement outcome, and the classical instruction $(\langle\cvar, \bit\rangle)$ is also executed to store the measurement outcome.
The set of all {\em atomic instructions} is $\allInstructions=\allInstructions^C \cup \allInstructions^U \cup \allInstructions^M$ where $\allInstructions^C$,   $\allInstructions^U$, and $\allInstructions^M$ are the sets of all classical, unitary and measurement instructions respectively.

\subsection{Atomic instructions executed on noisy quantum hardware}
A {\em hardware specification} $\hardware$ provides, for every unitary and measurement instruction $\instruction\in\allInstructions^U\cup\allInstructions^M$, a quantum channel $\hardware(\instruction)$ to model the noise when instruction $\instruction$ is executed on hardware~$\hardware$.
For a hardware specification~$\hardware$, the execution of an atomic instruction $\instruction\in\allInstructions$ on $\hardware$ is defined by the probabilistic successor function 
$
    \sem{\instruction}_\hardware : \hstates \rightarrow \ensembles
$.
As classical instructions are unaffected by noise, 
$\sem{\instruction}_\hardware(\hstate) = 
\sem{\instruction}(\hstate)$ 
for $\instruction \in \allInstructions^C$ and all hybrid states $\hstate \in \hstates$.

Intuitively, a quantum channel $\hardware(\instruction)$ for a unitary instruction $\instruction=\langle U,\qvars\rangle$  represents a {\em completely positive trace-preserving (CPTP) map}  acting on the space of linear operators~\cite{qc_bible} which guarantees that ensembles preserve their overall probability mass.
CPTP maps can always be described by a finite set of linear operators called {\em Kraus operators} in $\allops{\allqv}$.
Each Kraus operator can be intuitively understood as a different error process.
More concretely, we define a quantum channel for a unitary instruction $\instruction=\langle U,\qvars\rangle$ as a finite set of noise operators (i.e., Kraus operators) that satisfy $\sum_{\op \in \hardware(\instruction)} \op^\dagger \cdot \op = \Id$.
A hardware specification is noise-free for unitary instructions if it maps every unitary instruction to $\{\Id\}$.
Formally, the semantics of the unitary instruction $\instruction$ is then defined for a hybrid state $\hstate$ by
$
\sem{\instruction}_\hardware(\hstate) = 
        \sum_{
            \op \in \hardware(\instruction)
        }
        \left\{\branchAt{\hstate}{}{\op\cdot U_{\qvars}}:~\norm{ \op\cdot U_{\qvars} \cdot\topure{\hstate}{\allqv}}^2\right\}, 
$
where $\norm{ \op\cdot U_{\qvars} \cdot\topure{\hstate}{\allqv}}^2$ is the probability that a quantum error $\op$ ``affects'' the state.

\paragraph{Noisy measurement instructions.}
The quantum channel $\hardware(\instruction)$ for a measurement instruction $\instruction=\langle\qvar, \cvar\rangle$ affects the classical outcome of the measurement.
As a result, the post-measurement quantum state and the classical outcome may not coincide.
Formally, the quantum channel $\channel{\instruction} = \hardware(\instruction)$ specifies four probability values $\channel{\instruction}(\bit,\bit')$ for $\bit,\bit' \in \{0,1\}$, where $\channel{\instruction}(\bit,\bit')$ is the probability that the outcome of a measurement is $\bit$ if the post-measurement state at $\qvar$  has collapsed to~$\ket{\bit'}$.
Note that $\channel{\instruction}(0,0)+\channel{\instruction}(1,0)=1$ and $\channel{\instruction}(1,1)+\channel{\instruction}(0,1)=1$.
A hardware specification that is noise-free for measurements maps every measurement instruction to a channel $\schannel$ that satisfies $\schannel(1,0) = \schannel(0,1) = 0$.
The semantics of the measurement instruction is formally given as 
$\sem{\instruction}_\hardware(\hstate) = \sum_{\bit, \bit' \in \{0,1\}} \left\{
\sem{\langle \cvar, \bit \rangle}\left(\branchAt{\hstate}{}{\projector_{\bit'}}\right) :~\channel{\instruction}(\bit, \bit')\cdot\norm{\projector_{\bit'}\cdot\topure{\hstate}{\allqv}}^2\right\}.$

\section{A Noise-aware Quantum Hoare Logic (\textnormal{\textbf{\textsc{nQHL}}})}
\label{sec:nqhl}
In this section, we present the logical foundations for noise-aware analysis of quantum programs. 
We first present a programming language, followed by a relatively complete proof system for partial correctness with assertions that define sets of ensembles with respect to programs that are executed on noisy hardware.

\subsection{Hybrid quantum-classical programming language}

\subsubsection{Boolean terms and expressions.} 
The syntax and semantics of Boolean terms and expressions are defined in Figure~\ref{fig:bexprs}.
A {\em Boolean term} is semantically a function from hybrid states to $\{0,1\}$, and a {\em Boolean expression} denotes a set of hybrid states.
\begin{figure}[t!]
    \centering
    \scriptsize
    \[
    \begin{array}{rclcrcl}
        \BTerm \quad \bterm & := & \cvar \mid 0 \mid 1 & \quad\quad\quad & 
        \BExpr \quad \bexpr & := & \bterm_0=\bterm_1 \mid \bexpr_0 \land  \bexpr_1 \mid \neg \bexpr_0 \\
    \end{array}
    \]
    {\renewcommand{\arraystretch}{1.5}\begin{tabular}{|c|lll|}
    \hline
         signature & \multicolumn{3}{c|}{semantics} \\
    \hline
    {
        $\sem{\bterm}: \hstates \rightarrow \{0, 1\} $}
        & $\sem{\cvar}(\hstate) = \hstate(\cvar)$ 
        & $\sem{0}(\hstate) = 0$ 
        & $\sem{1}(\hstate) = 1$ 
        \\[2pt]
    \hline
    \multirow{2}{*}{$\sem{\bexpr} \subseteq \hstates $}
     & \multicolumn{3}{l|}{$\sem{\bterm_0 = \bterm_1} = \{\hstate \in \hstates \mid \sem{\bterm_0}(\hstate)=\sem{\bterm_1}(\hstate)\}$}
    \\[2pt]
     &$\sem{\neg\bexpr} = \hstates \setminus\sem{\bexpr}$&  \multicolumn{2}{l|}{\quad\quad$\sem{\bexpr_0 \land \bexpr_1} = \sem{\bexpr_0} \cap \sem{\bexpr_1}$} 
      \\[2pt]
      \hline
    \end{tabular}}
    \Description{Syntax  and semantics of Boolean expressions}
    \caption{Syntax  and semantics of Boolean expressions where $\cvar \in \allcv$ is a classical variable.} 
    \label{fig:bexprs}
\end{figure}

\noindent 
\paragraph{Notation.}
For an ensemble $\ensemble$ and a subset of hybrid states $S\subseteq\hstates$, the probability mass of hybrid states in $S$  is denoted by $\condProb{\ensemble}{S} = \sum_{\hstate \in S} \ensemble(\hstate)$. For a Boolean expression $\bexpr$, we similarly write $\condProb{\ensemble}{\bexpr} = \condProb{\ensemble}{\sem{\bexpr}}$ for the probability that the classical part of a hybrid state  in  $\supp(\ensemble)$ satisfies $\bexpr$.
The subensemble induced by $\bexpr$ is denoted by $\ensembleAt{\ensemble}{\bexpr} = \ensemble'$ and is defined for all hybrid states $\hstate'$ as $\ensemble'(\hstate') = \ensemble(\hstate')$ if $\hstate' \in \sem{\bexpr}$, and $\ensemble'(\hstate')=0$ otherwise.
The normalized subensemble induced by the set $S$  is denoted by $\branchAt{\ensemble}{}{S} = \zens$ if $\condProb{\ensemble}{S} = 0$. Otherwise $\branchAt{\ensemble}{}{S} = \ensembleAt{\ensemble}{S}/\condProb{\ensemble}{S}$. We write $\branchAt{\ensemble}{}{\bexpr}$ for $\branchAt{\ensemble}{}{\sem{\bexpr}}$.
Given a finite tuple $\cvars$ of classical variables and a classical state $\cstate$, we write $\cvars = \cstate$ for the Boolean expression $\bigwedge_{0 \leq i< |\cvars|}\cvars_i = \cstate(\cvars_i)$. We write $\branchAt{\ensemble}{\cvars}{\cstate}$ and $\branchAt{\ensemble}{\cvars}{\neg\cstate}$ for $\branchAt{\ensemble}{}{\cvars=\cstate}$ and $\branchAt{\ensemble}{}{\neg(\cvars=\cstate)}$ respectively.

\subsubsection{Programming language syntax and semantics.}
The programming language we consider extends  $\textsf{QIMDP}$~\cite{deng22} with noise-aware semantics using ensembles instead of density matrices.
Figure~\ref{fig:pl} shows the syntax and semantics; the semantics of a program $\prog$ on hardware $\hardware$ is a transformer $\sem{\prog}_{\hardware} : \allensZ{\hstates} \rightarrow \allensZ{\hstates}$ of subensembles.
\begin{figure}[t!]
    \centering
    \scriptsize
    \[
\begin{array}{rcl}
\texttt{Atomic quantum program} \quad \atomprog & := & \Skip \mid \cvar = \bterm \mid \cvar = \Measure (\qvar) \mid U(\qvars)\\
\texttt{Quantum program }\quad \prog & := & \atomprog \mid \prog_0;\prog_1 \mid \IfElse{\bexpr}{\prog_0}{\prog_1} \mid
                        \MyWhile{\bexpr}{\prog_0}
\end{array}
\]
\begingroup \scriptsize
    \[
\begin{array}{l}
\sem{\Skip}_{\hardware}(\ensemble) 
    = \ensemble 
\quad\quad
\sem{\cvar = \bterm}_{\hardware}(\ensemble) 
    = \sum_{\hstate \in \supp{\ensemble}}~\ensemble(\hstate) \cdot \sem{\langle \cvar, \sem{\bterm}(\hstate) \rangle}(\hstate) 
\\[4pt]
\sem{U(\qvars)}_{\hardware}(\ensemble) 
    = \sum_{\hstate \in \supp{\ensemble}} 
        \ensemble(\hstate)\cdot 
        \sem{\langle U, \qvars \rangle}_{\hardware}(\hstate)
\quad\quad
\sem{\cvar = \Measure(\qvar)}_{\hardware}(\ensemble) 
    = \sum_{\hstate \in \supp{\ensemble}} 
        \ensemble(\hstate)\cdot 
        \sem{\langle \qvar, \cvar \rangle}_{\hardware}(\hstate)
\\[4pt]
\sem{\IfElse{\bexpr}{\prog_0}{\prog_1}}_{\hardware}(\ensemble)
    = \sem{\prog_0}_{\hardware}(\ensembleAt{\ensemble}{\bexpr})
      + \sem{\prog_1}_{\hardware}(\ensembleAt{\ensemble}{\neg\bexpr})
      \quad\quad
\sem{\prog_0;\prog_1}_{\hardware}(\ensemble)
    = \sem{\prog_1}_{\hardware}(\sem{\prog_0}_\hardware(\ensemble))
\\[4pt]    
\sem{\MyWhile{\bexpr}{\prog}}_{\hardware}
    = {\rm lfp}~F, 
    \quad \text{where } 
    F(X)(\ensemble) 
      =  X\left(\sem{\prog}_{\hardware}\left(\ensembleAt{\ensemble}{\bexpr}\right)\right) 
        + \ensembleAt{\ensemble}{\neg\bexpr}
\end{array}
\]\endgroup
    \caption{Programming language syntax and semantics for a  hardware specification~$\hardware$.} 
    \label{fig:pl}
\end{figure}
We denote the set of all programs and the set of all {\em loop-free} programs (without {\tt while}) by~$\progs$ and $\finprogs$ respectively.
The set of all atomic programs is denoted by $\atomprogs$.

\subsection{Ensemble assertions}
\label{subsec:qhl}
\begin{figure}[t]
\centering
\scriptsize
\[
\begin{array}{rcl}
\QTerm \quad \qterm & := & \op \mid \ket{\qstate_0} \\
\StatesA \quad \stateA & := & \bexpr \mid \qterm_0 = \qterm_1 \mid \qterm = \ket{\qstate_1}_{\qvars} \mid  \neg \stateA_0 \mid  \stateA_0 \land \stateA_1 \\
\PTerm \quad \pterm & := & \real\cdot\probOf{\stateA, \op} \mid \pterm_0 + \pterm_1 \\
\EnsA \quad \ensA & := & \pterm_0 \geq \pterm_1 \mid \neg \ensA \mid \ensA_0 \land \ensA_1
\end{array}
\]
{\renewcommand{\arraystretch}{1.5}
    \scriptsize
    \centering
    \begin{tabular}{c|l}
        \hline
    signature & \multicolumn{1}{c}{semantics} \\
    \hline
    \multirow{1}{*}{$\sem{\qterm}: \hstates \rightarrow \hs_{\allqv} $}
     & $\sem{\op}(\hstate) = \op\cdot\topure{\hstate}{\allqv}$ \quad\quad $\sem{\ket{\qstate_0}}(\hstate) = \ket{\qstate_0}$\\
    \hline
    \multirow{4}{*}{$\sem{\stateA} \subseteq \hstates$}
    &
    $\sem{\qterm_0 = \qterm_1} = \left\{\hstate\in\hstates ~\big|~
    \frac{\sem{\qterm_0}(\hstate)}{\norm{\sem{\qterm_0}(\hstate)}} = 
    \frac{\sem{\qterm_1}(\hstate)}{\norm{\sem{\qterm_1}(\hstate)}} \right\}$
    \\
&$\sem{\qterm = \ket{\qstate_1}_{\qvars}} = 
    \left\{\hstate\in\hstates\mid (\ket{\qstate_1}\bra{\qstate_1})_{\qvars}\cdot \sem{\qterm}(\hstate) = \sem{\qterm}(\hstate)\right\}$
    \\   
     & {$\sem{\neg \stateA} = \hstates \setminus \sem{\stateA}$}\quad\quad {$\sem{\stateA_0 \land \stateA_1} = \sem{\stateA_0} \cap \sem{\stateA_1}$}\\
    \hline
    \multirow{2}{*}{$\sem{\pterm} : \ensembles \rightarrow \Reals$} & 
    {$\sem{\real\cdot \probOf{\stateA, \op}}(\ensemble) = \real\cdot \left(\sum_{\hstate \in \sem{\stateA}} \ensemble(\hstate) \cdot 
    \norm{\sem{\op}(\hstate)}^2 \right)$} 
    \\ &{$\sem{\pterm_0 + \pterm_1}(\ensemble) = \sem{\pterm_0}(\ensemble) + \sem{\pterm_1}(\ensemble)$}\\
    \hline % ensemble assertions
    \multirow{2}{*}{ $\sem{\ensA} \subseteq \ensembles$} & {$\sem{\pterm_0 \geq \pterm_1}=\{\ensemble\in\ensembles \mid \sem{\pterm_0}(\ensemble)\geq \sem{\pterm_1}(\ensemble)\}$}\\
    & $\sem{\neg \ensA} = \ensembles \setminus \sem{\ensA}$ \quad\quad  $\sem{\ensA_0 \land \ensA_1} = \sem{\ensA_0} \cap \sem{\ensA_1}$  \\
    \hline
    \end{tabular}}
    \Description{Syntax of terms and assertions.}
    \caption{Syntax and semantics of terms and assertions, where $\real \in \Reals$,~$\qvars$ is a tuple of distinct quantum variables in $\allqv$, and  $\ket{\qstate_0}$ and  $\ket{\qstate_1}$  are concrete constant vectors residing in the Hilbert spaces $\hs_{\allqv}$ and  $\hs_{\qvars}$, respectively. The matrix $\op$ is a linear operator acting on $\hs_{\allqv}$. State assertions may contain a Boolean expression $\bexpr$ defined in Figure~\ref{fig:bexprs}.} 
\label{fig:assertion_lang}
\end{figure}
\noindent
In Figure~\ref{fig:assertion_lang}, we define the syntax and semantics of ensemble assertions.
A {\em quantum term} $\qterm$ maps a hybrid state to a vector in $\hs_{\allqv}$.
A quantum term is specified by a matrix $\op$ acting on $\hs_{\allqv}$ or by a vector in $\hs_{\allqv}$.
The quantum term $\qterm=\Id$  denotes the quantum state of a hybrid state.
A {\em state assertion} $\stateA$ describes sets of hybrid states that satisfy a Boolean expression and equalities of quantum terms. 
The equality between quantum terms compares vectors residing in $\hs_{\allqv}$, whereas the second enables reasoning about subsystems and uses projectors of a concrete vector $\ket{\qstate_1}_{\qvars}$ to determine equality. 
Only the first comparison, between quantum terms, is phase-sensitive.
Given state assertions $\stateA_0$ and $\stateA_1$ we write $\stateA_0 \lor \stateA_1$ for $\neg(\neg \stateA_0 \land \neg \stateA_1)$.
A {\em probabilistic term} $\pterm$ maps an ensemble to a real number. 
Given an ensemble $\ensemble$, the probabilistic term 
$\probOf{\stateA,\op}$ is the expected squared norm of $\op$ applied to the states satisfying $\stateA$; when $\op$ is a projector, this is the probability of the corresponding measurement outcome.
We write $\probOf{\stateA}$ for $\probOf{\stateA, \Id}$, and if $\stateA = \top$ we write $\real$ instead of $\real\cdot\probOf{\top}$.
Finally, an {\em ensemble assertion} $\ensA$ describes sets of ensembles.
We write $\ensA_0 \lor \ensA_1$ for $\neg (\neg\ensA_0 \land \neg\ensA_1)$, and for an ensemble $\ensemble$, we write $\ensemble \vDash \ensA$ if $\ensemble \in \sem{\ensA}$. If $\sem{\ensA_0} \subseteq \sem{\ensA_1}$ ($\ensA_0$ entails $\ensA_1$), we write $\ensA_0 \Rightarrow \ensA_1$.

\begin{definition}[Validity]
    A Hoare triple 
    $\htriple{\pre}{\prog}{\post}{\hardware}$, 
    consisting of two ensemble assertions $\pre$ (the precondition) and $\post$ (the postcondition), a quantum program $P$, and a hardware specification $H$, 
    is valid iff for all ensembles $\ensemble,\ensemble' \in \ensembles$, 
    if $\ensemble \vDash \pre$ and $\ensemble'=\sem{\prog}_\hardware(\ensemble)$, then $\ensemble' \vDash \post$.    
\end{definition}
\paragraph{Syntactic substitutions.} In the remainder of this section we fix the quantum hardware $\hardware$.
Given an ensemble assertion $\ensA$, a classical variable $\cvar$, and a Boolean term $\bterm$, we have $\csubs{\ensA}{\cvar}{\bterm}$ denotes a {\em classical substitution} that  replaces every occurrence of $\cvar$ by $\bterm$.
Given a linear operator $\op$ for $\hs_{\allqv}$, the {\em quantum substitution} $\qsubs{\ensA}{\qvars}{\op}$ replaces every quantum term of the form $\qterm = \op'$ by $\qsubs{\qterm}{\qvars}{\op} = \op'\cdot \op$.
Given a Boolean expression $\bexpr$, we have $\rcond{\ensA}{\bexpr}$ denotes a {\em state assertion substitution}, which replaces every state assertion $\stateA$ by $\rcond{\stateA}{\bexpr} = \stateA\land \bexpr$.
Given a probabilistic term $\pterm$, we define a syntactic substitution for each quantum program inductively in Equation~\ref{eq:atom_pterm_subs}. 
{\scriptsize\begin{equation}
    \subsfunc{\hardware}{\prog}{\pterm} = \begin{dcases*}
          \pterm & if $\prog = \Skip$ \\
          \csubs{\pterm}{\cvar}{\bterm} & if $\prog =(\cvar= \bterm)$ \\ 
          \sum_{\bit, \bit' \in \{0,1\}}\hardware(\langle\qvar, \cvar\rangle)(\bit, \bit')\cdot \qsubs{\csubs{\pterm}{\cvar}{\bit}}{[\qvar]}{\projector_{\bit'}}
          & if $\prog = (\cvar = \Measure(\qvar))$ \\
          \sum_{\op\in \hardware(\langle U, \qvars \rangle)}\qsubs{\pterm}{\mvals{\allqv}}{\op\cdot U_{\qvars}} & if $\prog = U(\qvars)$ \\ 
          \rcond{\subsfunc{\hardware}{\prog_0}{\pterm}}{\bexpr} + ~\rcond{\subsfunc{\hardware}{\prog_1}{\pterm}}{\neg\bexpr} & if $\prog = \IfElse{\bexpr}{\prog_0}{\prog_1}$\\
          \subsfunc{\hardware}{\prog_0}{\subsfunc{\hardware}{\prog_1}{\pterm}} & if $\prog = \prog_0;\prog_1$\\
          \sum_{i=0}^{\infty} \Phi_\hardware^i(\rcond{\pterm}{\neg\bexpr})
          & if $\prog = \MyWhile{\bexpr}{\prog_0}$,
    \end{dcases*}
    \label{eq:atom_pterm_subs}
\end{equation}}
where $\Phi_\hardware^0(\pterm) = \pterm$ and for all $i > 0$, we have $\Phi_\hardware^i(\pterm) = \rcond{\subsfunc{\hardware}{\prog_0}{\Phi_\hardware^{i-1}(\pterm)}}{\bexpr}.$
The proof of Lemma~\ref{lemma:dual_pterm} is in Appendix Section~\ref{appendix:dual_pterm}.

\begin{lemma}[Probabilistic substitutions] 
For all probabilistic terms $\pterm$, ensembles $\ensemble \in \ensembles$, quantum programs $\prog \in \progs$, and quantum hardware specifications $\hardware$, 
if $\ensemble' = \sem{\prog}_\hardware(\ensemble)$, then $\sem{\subsfunc{\hardware}{\prog}{\pterm}}(\ensemble) = \sem{\pterm}(\ensemble')$.
\label{lemma:dual_pterm}
\end{lemma}

\begin{definition}[Weakest liberal precondition] For all ensemble assertions $\post$,  quantum hardware specifications $\hardware$, and quantum programs $\prog \in \progs$, 
the weakest liberal precondition of $\prog$ with respect to $\post$ is an assertion
$\pre$ such that (i)~the triple $\{\pre \}~\prog~\{\post\}$ is valid and (ii)~for all assertions $\pre'$, if $\{ \pre' \}~\prog~\{\post\}$ is valid, then $\pre' \Rightarrow \pre$.
\label{def:wp}
\end{definition}

\noindent
The proof system shown in  Figure~\ref{fig:proof_system} uses the function $\weakp{\hardware}{\prog}{\post}$, which, for hardware $\hardware$, maps a program $\prog$ and a postcondition $\post$ to the ensemble assertion:
\begin{equation}
    \weakp{\hardware}{\prog}{\post} = \begin{dcases*}
        \subsfunc{\hardware}{\prog}{\pterm_0}\geq\subsfunc{\hardware}{\prog}{\pterm_1} & if $\post = \pterm_0\geq\pterm_1$,\\
        \neg\weakp{\hardware}{\prog}{\post'}   & if  $\post = \neg \post'$,\\
        \weakp{\hardware}{\prog}{\post_0}\land\weakp{\hardware}{\prog}{\post_1} &  if $\post = \post_0\land\post_1$.
    \end{dcases*}
\end{equation}
Relative completeness for partial correctness is shown in Theorem~\ref{thm:proof_system} (proof in Appendix Section~\ref{appendix:thm:proof_system}).

\begin{figure}[t]
    \tiny
    \[
    \begin{array}{cc}
    \inferrule*[right=Atom]{
            \prog \in \atomprogs
        }{
            \htriple{\weakp{\hardware}{\prog}{\post}}{\prog}{\post}{\hardware}
        } 
    &
        \inferrule*[right=Cond]{
        }{
            \htriple{\weakp{\hardware}{(\IfElse{\bexpr}{\prog_0}{\prog_1})}{\post}}{(\IfElse{\bexpr}{\prog_0}{\prog_1})}{\post}{\hardware}
        }
    \end{array}
    \]
    \[
        \begin{array}{cc}
        \inferrule*[right=Seq]{
            \htriple{\pre_0}{\prog_0}{\post_0}{\hardware}
            \quad 
            \htriple{\post_0}{\prog_1}{\post_1}{\hardware}
        }{
            \htriple{\pre_0}{(\prog_0;\prog_1)}{\post_1}{\hardware}
        } 
        & \inferrule*[right=While]{
        }{
            \htriple{
            \weakp{\hardware}{(\MyWhile{\bexpr}{\prog})}{\post}
            }{(\MyWhile{\bexpr}{\prog})}{\post}{\hardware}
        } 
        \end{array}
    \]
    \[
    \begin{array}{cc}
        &
        \inferrule*[right=conseq]{
        \pre \Rightarrow \pre'  \quad\quad \htriple{\pre'}{\prog}{\post'}{\hardware}\quad\quad \post' \Rightarrow \post
        }{\htriple{\pre}{\prog}{\post}{\hardware}}
    \end{array}
    \]
    \Description{Proof system}
    \caption{Proof system.}
    \label{fig:proof_system}
\end{figure}

\begin{theorem}[Proof system] Let $\vdash$ denote syntactic derivability using the proof system of Figure~\ref{fig:proof_system}. For all quantum programs $\prog \in \progs$, hardware specifications $\hardware$, and ensemble assertions $\pre, \post$, the following hold:
\begin{enumerate}
    \item \textbf{Soundness.} If $\vdash \htriple{\pre}{\prog}{\post}{\hardware}$,  then $\htriple{\pre}{\prog}{\post}{\hardware}$ is valid.
    \item \textbf{Weakest liberal precondition derivability.} If $\pre$ is the weakest liberal precondition for $\prog$ and $\post$, then 
    $\vdash \htriple{\pre}{\prog}{\post}{\hardware}$.
    \item \textbf{Relative completeness.} If $\htriple{\pre}{\prog}{\post}{\hardware}$ is valid, then $\vdash\htriple{\pre}{\prog}{\post}{\hardware}$ relative to an oracle deciding $\Rightarrow$ (entailment of assertions).
\end{enumerate}
\label{thm:proof_system}
\end{theorem}

\begin{example}[Local phase  preparation proof] Recall   Example~\ref{ex:motivation}.
The weakest liberal  precondition of $\post = \probOf{\Id = -\ket{1}} \geq 0.5$ and quantum program $\prog = \Z(\qvar)$ is 
$$\weakp{\hardware_0}{\prog}{\post} = \probOf{\X_\qvar\Z_\qvar = -\ket{1}, \sqrt{0.1}\X_\qvar\Z_\qvar} + \probOf{\Z_\qvar = -\ket{1}, \sqrt{0.9}\Z_\qvar}  \geq 0.5.$$
Equivalently, $\weakp{\hardware_0}{\prog}{\post} = 0.1\cdot\probOf{ \Id = -\ket{0}} + 0.9\cdot\probOf{\Id = \ket{1}}  \geq 0.5$.
Thus the triple in  Equation~\ref{eq:mot_example} is invalid for $\hardware_0$ since 
its precondition does not imply  $\weakp{\hardware_0}{\prog}{\post}$.
\label{ex:wp}
\end{example}
 
\begin{figure}[t]
    \centering
    \scriptsize
    \[
      \begin{array}{l}
        \inferrule*{ % premises
            \prog \in \atomprogs
        }{ 
            \langle \prog, \ensemble \rangle \rightarrow_1^{\hardware,\atomprogs} \langle \halt, \sem{\prog}_\hardware(\ensemble) \rangle
        }
      \end{array}
      \begin{array}{l}
        \inferrule*{ % premises
            \langle \prog_0, \ensemble\rangle  \rightarrow_\prob^{\hardware,\atomprogs} \langle\prog_0', \ensemble'\rangle
        }{ 
            \langle \prog_0;\prog_1, \ensemble \rangle \rightarrow_{\prob}^{\hardware,\atomprogs} \langle \prog_0';\prog_1, \ensemble'\rangle
        }
      \end{array}
      \begin{array}{l}
        \inferrule*{ % premises
            \prog_0 = \halt
        }{ 
            \langle \prog_0;\prog_1, \ensemble \rangle \rightarrow_{1}^{\hardware,\atomprogs} \langle \prog_1, \ensemble\rangle
        }
      \end{array}
      \]
      % sequential rules
      % if else rules
    \[
      \begin{array}{l}
        \inferrule*{ % premises
            \ensemble_{\bexpr} = \branchAt{\ensemble}{}{\bexpr} \in \ensembles
        }{ 
            \langle   \IfElse{\bexpr}{\prog_0}{\prog_1}, \ensemble \rangle \rightarrow_{\condProb{\ensemble}{\bexpr}}^{\hardware,\atomprogs} \langle \prog_0, \ensemble_\bexpr\rangle
        }
      \end{array}
      \begin{array}{l}
        \inferrule*{ % premises
             \ensemble_{\neg \bexpr} = \branchAt{\ensemble}{}{\neg\bexpr} \in \ensembles
        }{ 
            \langle  \IfElse{\bexpr}{\prog_0}{\prog_1}, \ensemble \rangle \rightarrow_{\condProb{\ensemble}{\neg\bexpr}}^{\hardware,\atomprogs} \langle \prog_1, \ensemble_{\neg \bexpr}\rangle
        }
      \end{array}
    \]
    \caption{Operational rules for executing a quantum program on a hardware specification~$\hardware$.}
    \Description{Operational rules}
    \label{fig:operational_rules}
\end{figure}

\section{Numerically Precise Model Checking of Loop-free Quantum Programs on Noisy Hardware}
\label{sec:verification}
Hoare logic provides a sound and compositional characterization of valid triples;
however, it does not induce an effective decision procedure.
Furthermore, the presence of noise makes correctness a property that must be analyzed for different hardware specifications, ideally automatically.
To address this limitation, in this section we develop algorithms for the noise-aware verification of quantum programs, enabling full automation in the loop-free case.

\begin{problem}[Noise-aware verification]
  Given a hardware specification $\hardware$, quantum program $\prog \in \progs$, and ensemble assertions $\pre$ and $\post$, decide if the hardware-dependent Hoare triple $\htriple{\pre}{\prog}{\post}{\hardware}$ is valid.
\end{problem}

\noindent
While the verification problem is undecidable in general~\cite{church36},
we give a decision procedure for loop-free programs and linear preconditions, 
which is based on exploring a finite part of an infinite graph whose nodes are configurations.
A {\em configuration} is a pair $\config  = \langle \prog, \ensemble \rangle$ consisting of a loop-free program $\prog \in (\finprogs \cup \{\halt\})$ and an ensemble $\ensemble \in \ensembles$, and
denote by $\configs$ the set of all configurations.

\begin{definition}[Ensemble graph]  
The {\em ensemble graph} $EG(\hardware, \atomprogs)=\langle \configs, \QMCTransitionF \rangle$ of a hardware specification $\hardware$ over the set of atomic programs is
an infinite Markov chain whose states are the configurations $\configs$, and 
whose stochastic transition function $\QMCTransitionF: \configs \times \configs \rightarrow [0,1]$ is defined for all configurations $\config, \config' \in \configs$ by $\QMCTransitionF(\config, \config')= \prob$ iff $\config \rightarrow_\prob^{\hardware,\atomprogs} \config'$ can be derived using the operational rules from Figure~\ref{fig:operational_rules}.  
Otherwise, $\QMCTransitionF(\config, \config') = 0$.
\end{definition}

\noindent
We denote the {\em depth} of a loop-free quantum program $\prog$ with respect to the set $\atomprogs$ of atomic programs  by $\plength{\prog}{\atomprogs}$, and define it inductively by
\begin{equation*}
    \plength{\prog}{\atomprogs} = \begin{dcases*}
        1 &if $\prog\in\atomprogs$, \\
        \plength{\prog_0}{\atomprogs} + \plength{\prog_1}{\atomprogs} &if $\prog =\prog_0;\prog_1$, \\
        \max(\plength{\prog_0}{\atomprogs}, \plength{\prog_1}{\atomprogs}) & 
        if $\prog = \IfElse{\bexpr}{\prog_0}{\prog_1}$.
    \end{dcases*}
\end{equation*}
For $\prog$ and an initial ensemble $\ensemble$, 
we explore the finite part of $EG(\hardware, \atomprogs)$ that is reachable from the configuration $\langle \prog,\ensemble\rangle$.
For a configuration $\config\in\configs$, 
the set of configurations that are {\em reachable from $\config$} in the ensemble graph $\MC$ with transition function $\QMCTransitionF$ is the smallest set $\ReachConfs(\MC, \config)$ such that
(i)~$\config\in\ReachConfs(\MC, \config)$, and
(ii)~if $\config_0\in\ReachConfs(\MC, \config)$, then every successor of $\config_0$ is also in $\ReachConfs(\MC, \config)$, i.e., $\{\config_1 \in \configs \mid (\config_0, \config_1) \in \supp(\QMCTransitionF) \}\subseteq \ReachConfs(\MC, \config)$.

\subsection{Verification of Hoare triples with singular preconditions}
Let $\mathbb{Q}$ denote the set of all rational numbers.
An ensemble $\ensemble\in \ensembles$ is 
{\em representable} if 
$\ensemble(\hstate)\in\mathbb{Q}$ for all 
$\hstate\in\supp(\ensemble)$.
For a representable ensemble $\ensemble$, let
{$\pre_\ensemble = \left(\bigwedge_{\hstate \in \supp(\ensemble)} \probOf{\stateA_{\hstate}} = \ensemble(\hstate) \right)$} be the ensemble assertion that is satisfied by $\ensemble$, and $\ensemble$ only; where $\stateA_{\hstate} = \left(\Id = \topure{\hstate}{\allqv} \land \mvals{\allcv} = \topure{\hstate}{\allcv} \right)$ denotes the state assertion satisfied by the hybrid state $\hstate \in \hstates$, and $\hstate$ only.
An ensemble assertion $\pre$ is {\em singular} if there exists a representable ensemble $\ensemble$ such that $\pre = \pre_\ensemble$.
We denote this unique ensemble by $\ensemble_\pre$.
%\end{definition}
Given an ensemble graph $\MC = \langle\configs, \QMCTransitionF\rangle$ and a configuration $\config= \langle\prog, \ensemble\rangle$, we define the function $\getFinalEnsemble(\MC, \config)$ to compute the final ensemble reached by the loop-free program $\prog$ from the initial ensemble $\ensemble$  by $\getFinalEnsemble(\MC, \config) = \ensemble$ if $\prog = \halt$. Otherwise, $\getFinalEnsemble(\MC, \config) = \sum_{(\config, \config') \in \supp(\QMCTransitionF)} \QMCTransitionF(\config, \config') \cdot \getFinalEnsemble(\MC, \config')$.

\begin{lemma}[Final ensemble] 
For all loop-free quantum programs $\prog$, 
hardware specifications $\hardware$, and ensembles $\ensemble$,  
 $\getFinalEnsemble(EG(\hardware, \atomprogs), \langle\prog, \ensemble \rangle) = \sem{\prog}_\hardware(\ensemble)$. 
 Furthermore, if $\ensemble$ is a representable ensemble, then  $\getFinalEnsemble(EG(\hardware, \atomprogs), \langle\prog, \ensemble \rangle)$ can be computed exactly using $\mathcal{O}(\plength{\prog}{\atomprogs})$ precision units (e.g., bits).
 \label{lemma:final-ensemble}
\end{lemma}

\noindent
The proof of Lemma~\ref{lemma:final-ensemble} follows by induction on the program $\prog$ (see Appendix, Section~\ref{appendix:proof-final-ensemble}).
Thus, to solve the verification problem for a loop-free quantum program $\prog$,
hardware $\hardware$, singular precondition $\pre$, and postcondition $\post$,
it remains to decide that the final ensemble 
$\sem{\prog}_\hardware(\ensemble_\pre)$ satisfies $\post$.
For this purpose, 
we translate a given representable ensemble $\ensemble$ and ensemble assertion $\post$ into a first-order formula $\textsc{realForm}(\ensemble,\post)$ that is 
satisfiable in the theory of real-closed fields iff 
$\ensemble\in\sem{\post}$.
The function $\textsc{realForm}$ takes two arguments, a symbolic ensemble and an ensemble assertion.
A {\em symbolic ensemble} $\gamma$ is a finite set of pairs $(\hstate,w)$ consisting of 
a hybrid state $\hstate$ and a symbolic expression~$w$,
such that each hybrid state $\hstate$ occurs only once in~$\gamma$.
Let $\gamma(\hstate)$ be $w$ if $(\hstate,w)\in\gamma$, and $\gamma(\hstate)=0$ otherwise.
A {\em symbolic expression} is either a rational constant $\real\in\mathbb{Q}$,
or a variable $u$ ranging over the reals,
or a sum or product of symbolic expressions.
Note that every representable ensemble is a symbolic ensemble all of whose symbolic expressions are rational constants.
Before defining the function {\sc realForm}, 
we define a function {\sc realTerm}
that maps a symbolic ensemble $\gamma$ and a probabilistic term $\pterm$
to a symbolic expression.
Let $S(\gamma) = \{\hstate \mid \gamma(\hstate) \neq 0\}$ be the (finite) set of hybrid states that have a nonzero or symbolic probability in $\gamma$. 
Then $\textsc{realTerm}(\gamma,\pterm)$ and $\textsc{realForm}(\gamma,\post)$ are inductively defined as follows:

\begin{itemize}
    \item If $\pterm = \real\cdot \probOf{\stateA, \op}$, let $S' = \{ \hstate \in  S \mid \hstate \in \sem{\stateA}\}$. Then 
    $\textsc{realTerm}(\gamma,\pterm)=\sum_{\hstate \in S'} \real\cdot\gamma(\hstate)\cdot \norm{\sem{\op}(\hstate)}^2$.
    \item If $\pterm = \pterm_0 + \pterm_1$, 
    then $\textsc{realTerm}(\gamma,\pterm)=\textsc{realTerm}(\gamma,\pterm_0) + \textsc{realTerm}(\gamma,\pterm_1)$.
    \item If $\post = \pterm_0 \geq \pterm_1$, then $\textsc{realForm}(\gamma,\post)=\textsc{realTerm}(\gamma,\pterm_0) \geq \textsc{realTerm}(\gamma,\pterm_1)$.
    \item If $\post = \neg\post'$, then $\textsc{realForm}(\gamma,\post)=\neg\textsc{realForm}(\gamma,\post')$.
    \item If $\post = \post_0 \land \post_1$, then $\textsc{realForm}(\gamma,\post)=\textsc{realForm}(\gamma,\post_0) \land \textsc{realForm}(\gamma,\post_1)$.
\end{itemize}
The following lemma follows by induction on~$\post$ (see Appendix Section~\ref{sec:proof-real-formula}).

\begin{lemma}[First-order formula] For all representable ensembles $\ensemble$ and all ensemble assertions $\post$, the first-order formula \textnormal{\textsc{realForm}}$(\ensemble, \post)$
is satisfiable in the theory of real-closed fields iff $\ensemble\vDash \post$.
\label{lemma:real-formula}
\end{lemma}

\begin{example}[Verification] Consider again the motivating Example~\ref{ex:motivation} from the introduction.
Let $\pre_\ensemble$ be the singular precondition  satisfied only by the ensemble  $\ensemble = \{(\ket{0}, 0): 0.5,~(\ket{1}, 0): 0.5\}$ and denote by $\post$ the postcondition from Equation~\ref{eq:mot_example}.
Let $\ensemble_0 = \getFinalEnsemble(EG(\hardware_0, \atomprogs), \langle\instruction, \ensemble\rangle)$ and 
$\ensemble_1 = \getFinalEnsemble(EG(\hardware_1, \atomprogs), \langle\instruction, \ensemble\rangle)$ be the final ensembles reached when executing $\instruction$ on $\hardware_0$ and $\hardware_1$, respectively.
Since $\ensemble_0((-\ket{1}, 0)) = 0.45$ and $\ensemble_0((-\ket{1}, 1)) = 0$,
our tool determines in less than 0.002 seconds that the triple $\htriple{\pre_{\ensemble}}{\instruction}{\post}{\hardware_0}$ is invalid  because 
$\textsc{realForm}(\ensemble_0, \post) = 0.45 \geq 0.5$ is not satisfiable. In contrast, $\htriple{\pre_{\ensemble}}{\instruction}{\post}{\hardware_1}$ is valid because $\textsc{realForm}(\ensemble_1, \post) = 0.5 \geq 0.5$ holds.
\label{ex:auto_verif}
\end{example}

\subsection{Verification of Hoare triples with linear preconditions}
\label{subsec:verify_lin}
We generalize the verification procedure for singular preconditions to linear preconditions.
A \emph{polytope} over ensembles is the set of convex combinations of a finite set 
$
\polytope = \{\ensemble_0, \dots, \ensemble_{N-1}\}
$
of ensembles:
$\conv(\polytope) \;=\; 
\{ \sum_{i=0}^{N-1} \prob_i \cdot \ensemble_i 
\;\mid\; \forall 0\leq i< N.\ \prob_i \in \Reals \;\wedge\; \prob_i\geq0 \; \land\;\sum_{i=0}^{N-1} \prob_i = 1 \}$.
Each ensemble $\ensemble_i$ is a \emph{corner} of the polytope $\conv(\polytope)$.
Given a set $S$ of hybrid states, denote by $\stateA_S = \bigvee_{\hstate \in S} \stateA_{\hstate}$ the state assertion that represents $S$ exactly, i.e., $\sem{\stateA_S} = S$.
The proof of Lemma~\ref{lemma:lin_prec_formula} is in the Appendix, Section~\ref{appendix:lemma:lin_prec_formula}.

\begin{lemma}[Polytope assertion]
Given a finite set $\polytope \subset \ensembles$ of ensembles, let $S = \bigcup_{\ensemble \in \polytope} \supp(\ensemble)$, and for an ensemble $\ensemble \in \polytope$, write $\ket{\ensemble} \in \Reals^{|S|}$ to denote the ensemble (column) vector over the real vector space $\Reals^{|S|}$ with an orthonormal basis indexed by the hybrid states $\hstate \in S$,  i.e.,  $\braket{\hstate\mid\ensemble}= \ensemble(\hstate)$ for all $\hstate \in S$.
If the set of vectors $\mvals{\polytope}=\{\ket{\ensemble} \mid \ensemble \in \polytope\}$ is a linearly independent set of vectors, define the polytope assertion for  $\polytope$ by 
$\pre_{\polytope} = \left(\probOf{\stateA_S} = 1
\;\land\; \bigwedge_{\ensemble \in \polytope} \pterm_\ensemble \geq 0
\;\land\; \bigwedge_{\hstate \in S} \Big(\probOf{\stateA_\hstate} = \sum_{\ensemble\in\polytope} \ensemble(\hstate)\cdot\pterm_\ensemble\Big)\right),$
where $\pterm_\ensemble = \sum_{\hstate\in S} L_{\ensemble,\hstate}\cdot\probOf{\stateA_\hstate}$ is a probabilistic term using  the left pseudo-inverse $L$  of the matrix $V \in \Reals^{|S|\times|\polytope|}$ with elements $V_{\hstate, \ensemble} = \ensemble(\hstate)$.
For all ensembles $\ensemble \in \ensembles$, we have $\ensemble \in \conv(\polytope)$ iff $\ensemble \vDash \pre_{\polytope}$.
\label{lemma:lin_prec_formula}
\end{lemma}
It follows that if $\mvals{\polytope}$ is  linearly dependent, its polytope assertion can be defined as a disjunction over all maximal linearly independent subsets of $\mvals{\polytope}$~\cite{schrijver86}. 
Given a finite set $\QSetPolytopes$ of polytopes, we write $\pre_{\QSetPolytopes} =  \bigvee_{\polytope \in \QSetPolytopes}\pre_{\polytope}$.
An ensemble assertion $\pre$ is {\em linear} if there exists a finite set of polytopes $\QSetPolytopes$ such that $\pre = \pre_{\QSetPolytopes}$.

Let $\prepolys(\pre) = \QSetPolytopes$ denote the finite set of polytopes of a   linear ensemble assertion $\pre$ and consider the ensemble graph $\MC = EG(\hardware, \atomprogs)$.
We define a procedure
$\textsc{checkLin}(\MC,\prog,\pre,\post)$
which decides the validity of the Hoare triple $\htriple{\pre}{\prog}{\post}{\hardware}$ for loop-free programs $\prog$ and linear preconditions $\pre$.
Here we give only an outline (see {Algorithm}~\ref{alg:checkLin} in the Appendix, Section~\ref{appendix:checklin_alg}).
For each polytope $\conv(\polytope)\in\prepolys(\pre)$, $\textsc{checkLin}$ computes the set $\polytope' = \{\getFinalEnsemble(\MC, \langle \prog,\ensemble_i\rangle) \mid  \ensemble_i \in \polytope\} = \{\ensemble_0', \ldots, \ensemble_{N-1}'\}$ of final ensembles reached from each corner of the polytope.
Since every reachable ensemble from $\conv(\polytope)$ is a convex combination of the ensembles in $\polytope'$, we encode a symbolic ensemble $\gamma = \sum_{i=0}^{N-1}u_i\cdot \ensemble_i'$ using fresh variables $u_0,\ldots,u_{N-1}$ and return $\texttt{false}$ if $\forall u_0,\ldots,u_{N-1}.\ \Big(\bigwedge_{i=0}^{N-1} u_i \geq 0 \land \sum_{i=0}^{N-1} u_i = 1\Big) \Rightarrow \textsc{realForm}(\gamma, \post)$ is unsatisfiable. 
Otherwise, we repeat this process with the remaining  polytopes or return $\texttt{true}$ if there are no more polytopes.
The correctness of $\textsc{checkLin}$ follows by linearity of program semantics (see Appendix, Section~\ref{proof:linear-prec-verification}).
The precision required for verifying linear preconditions is bounded similarly to singular preconditions.

\begin{theorem}[Correctness of $\textsc{checkLin}$] For all hardware specifications $\hardware$, loop-free quantum programs $\prog$, and ensemble assertions $\pre$ and $\post$,
if $\pre$ is linear,
then $\htriple{\pre}{\prog}{\post}{\hardware}$ is valid iff \textnormal{\textsc{checkLin}}$(EG(\hardware, \atomprogs), \prog, \pre, \post)=\texttt{true}$.
\label{thm:linear-prec-verification}
\end{theorem}

\section{Bounded Synthesis of Quantum Programs for Noisy Hardware}
\label{sec:synthesis}
{While in the previous section we considered the bounded-horizon verification of quantum programs, in this section we consider the automatic synthesis of loop-free quantum programs.
In particular, given a precondition and a postcondition, we synthesize a loop-free quantum program  with respect to a hardware specification and an instruction set, such that the number of executed instructions is bounded by a given horizon.}
An {\em instruction set} $\synthprogs$ is a set of loop-free quantum programs which may contain atomic instructions as well as more complex constructs. 
In our experiments, we use the following instruction sets (different instructions are separated by commas):
\begingroup \scriptsize
\begin{equation}
\begin{aligned}
&\CXHIns = \{
  \Had([\qvar_1]),\;\;  
  \Had([\qvar_2]),\;\;  
  \cvar = \Measure(\qvar_2),\;\;
  \CX([\qvar_2, \qvar_1]),\;\;
  \CX([\qvar_0, \qvar_1])
\},\\
&\IPMAInsP = 
\{
  \cvar = \Measure(\qvar_2),\;\;
  \CX([\qvar_0, \qvar_2]);\;
  \CX([\qvar_1, \qvar_2]),\;\;
  \X([\qvar_0])
\},
\\
&\GHZIns = 
\{
\Had([\qvar_i]),\;\;
  \CX([\qvar_i, \qvar_j])\;\mid 0 \leq i,j \leq 2 \land i \neq j
\},
\\
&\ResetIns = 
\{
 \cvar = \Measure(\qvar_0),\;\;
 \X([\qvar_0])
\},
\quad\quad
\LPhaseIns =
\{
  \Z([\qvar_0]),\;\;
  \X([\qvar_0]),\;\;
  \cvar = \Measure(\qvar_1)
\}.
\end{aligned}
\label{eq:instruction_sets}
\end{equation}
\endgroup
Formally, these programs follow the grammar
$\begingroup 
    \cprog  =   \synthins \mid \cprog_0;\cprog_1 \mid  
    \IfElse{\bexpr}{\cprog_0}{\cprog_1} \mid \halt
\endgroup$ with $\synthins \in \synthprogs$, and are called \emph{candidate programs}.
We write $\cprogs_\synthprogs$ for the set of candidate programs over the instruction set $\synthprogs$.
The set of possible {\em next} instructions to be executed is defined inductively for a program $\prog$ and $\ensemble \in \ensembles$ by
\begingroup
\footnotesize
\begin{equation}
    \nextI(\prog, \ensemble) = \begin{dcases*}
        \emptyset & if $\prog = \halt$ or $\ensemble = \zens$\\
        \{\prog\} &if $\prog \in \synthprogs$, \\
        \nextI(\prog_0, \ensemble) &if $\prog  = \prog_0;\prog_1$, \\
        \nextI(\prog_0, \branchAt{\ensemble}{}{\bexpr}) \cup  \nextI(\prog_1, \branchAt{\ensemble}{}{\neg \bexpr})&if $\prog = \IfElse{\bexpr}{\prog_0}{\prog_1}$ \\
    \end{dcases*}
\end{equation}
\endgroup

\noindent
{\em Admissible programs.} 
An {\em instruction guard} for an instruction set $\synthprogs$ is a function $\guard$ that maps each instruction $\synthins\in \synthprogs$ to an ensemble assertion that specifies when $\synthins$ can be applied.
For example, in order to preserve the entanglement of certain qubits, 
an instruction guard can be used to prevent synthesized programs from performing destructive measurements.
Given an instruction set $\synthprogs$ with an instruction guard $\guard$ for $\synthprogs$,
a candidate program $\prog\in\cprogs_\synthprogs$ is {\em admissible} in an ensemble $\ensemble\in\ensembles$
iff for all reachable configurations $\langle\prog', \ensemble'\rangle \in \ReachConfs(EG(\hardware, \synthprogs), \langle\prog, \ensemble\rangle)$, 
and for all instructions $\synthins \in \nextI(\prog', \ensemble')$, we have  $\ensemble' \vDash \guard(\synthins)$.
The program $\prog$ is {\em admissible} with respect to a precondition $\pre$ iff $\prog$ is admissible in $\ensemble$ for all initial ensembles $\ensemble \vDash \pre$.
For linear preconditions $\pre$,
we can check admissibility by verifying that $\prog$ is admissible in all ensembles $\ensemble \in \conv(\polytope)$,
for all polytopes $\conv(\polytope) \in \prepolys(\pre)$.
We do this iteratively up to the depth of $\prog$ using the verification method from
Subsection~\ref{subsec:verify_lin} for linear preconditions. Note that the ensemble graph is defined using the instruction set  $\synthprogs$ instead of the set $\atomprogs$ of atomic programs. 

\begin{problem}[Noise-aware bounded synthesis] 
Given a hardware specification $\hardware$, 
an instruction set $\synthprogs$
with an instruction guard $\guard$,
a horizon $\horizon\ge 0$,
and two ensemble assertions $\pre$ and~$\post$,
find a candidate program $\prog\in\cprogs_\synthprogs$ 
such that (i)~$\plength{\prog}{\synthprogs} \leq \horizon$, (ii)~$\prog$ is admissible with respect to the precondition $\pre$, and 
(iii)~$\htriple{\pre}{\prog}{\post}{\hardware}$ is valid.
\label{problem:linear_program_synthesis}
\end{problem}

\subsection{Bounded synthesis for linear preconditions and arbitrary postconditions}
In the following, we fix a hardware specification $\hardware$, instruction set $\synthprogs$,
instruction guard $\guard$ for $\synthprogs$, and horizon $\horizon\ge 0$.
To synthesize a quantum program that satisfies a Hoare triple with linear precondition $\pre$ and arbitrary postcondition $\post$,
we first enumerate all candidate programs of depth at most~$\horizon$ which are ``dominant,''
and compute the final ensemble each program reaches when starting from any of 
the finitely many corners of the polytopes in $\prepolys(\pre)$. 
Finally, we verify if there exists a  dominant program $\prog$ such that  $\htriple{\pre}{\prog}{\post}{\hardware}$ is valid.

For the first step, define a finite set $\fincprogs$ of candidate programs. 
Since the horizon does not directly constrain the number of conditional instructions in a program, we choose one representative from each class of finitely many equivalence classes of branching programs with depth at most~$\horizon$.
For an ensemble $\ensemble$, let $\cstatesOf{\ensemble} = \{\topure{\hstate}{\allcv} \mid \hstate \in \supp(\ensemble)\}$
be the corresponding set of classical states. 
For an ensemble $\ensemble$ and $0\le l\le \horizon$, define the set $\DCandidates(\ensemble, l)$
of programs inductively, with $\DCandidates(\ensemble, 0) = \{\halt\}$ and 
{
\small
$
\DCandidates(\ensemble, l+1) = \DCandidates(\ensemble, l)\  \cup \left\{ \synthins; \prog ~\big|~~~ \synthins\in\synthprogs\land|\cstatesOf{\sem{\synthins}_\hardware(\ensemble)}| = 1 \land \prog \in \DCandidates(\sem{\synthins}_\hardware(\ensemble), l)\right\}
    \cup\left\{ \synthins; \prog ~\big|~~~ \synthins\in\synthprogs\land |\cstatesOf{\sem{\synthins}_\hardware(\ensemble)}| > 1\land \prog \in \CBlock(\sem{\synthins}_\hardware(\ensemble), l)  \right\}.
$
}
For an ensemble $\ensemble$ and $0\le l\le \horizon$, the auxiliary function $\CBlock(\ensemble, l)$ returns all possible conditional blocks that branch using only classical states in the support of $\ensemble$: if $|\cstatesOf{\ensemble}| = 1$, then $\CBlock(\ensemble, l) = \DCandidates(\ensemble, l)$. 
Otherwise, 
$\CBlock(\ensemble,l) = 
    \{
        \IfElse{\mvals{\allcv} = \cstate}{\prog_0}{\prog_1} \mid 
        \cstate \in \cstatesOf{\ensemble} \land \prog_0 \in  \DCandidates(\branchAt{\ensemble}{\mvals{\allcv}}{\cstate}, l) \land  
        \prog_1 \in \CBlock(\branchAt{\ensemble}{\mvals{\allcv}}{\neg \cstate}, l)
    \}
$.
The desired program is in the set $\fincprogs = \{\prog \in \bigcup_{\conv(\polytope) \in \prepolys(\pre)}\DCandidates(\sum_{\ensemble \in \polytope}\frac{1}{|\polytope|}\cdot\ensemble, \horizon)\mid \prog \text{ is admissible w.r.t } \pre\}.$
However, we avoid exhaustive enumeration of all programs in $\fincprogs$ by computing only dominant programs.
Given two candidate programs $\prog_0, \prog_1 \in \fincprogs$, we say that $\prog_1$ {\em dominates} $\prog_0$ with respect to a precondition~$\pre$ if each final ensemble for program $\prog_1$ is also a final ensemble for $\prog_0$.
Formally, let 
$\prog_0 \leq_{\pre} \prog_1$ if for all polytopes $\conv(\polytope) \in \prepolys(\pre)$ holds $\{\sem{\prog_1}_\hardware(\ensemble) \mid  \ensemble \vDash \pre_{\polytope}\} \subseteq  \{\sem{\prog_0}_\hardware(\ensemble) \mid  \ensemble \vDash \pre_{\polytope}\}$.
A program $\prog\in\fincprogs$ is $\phi$-{\em dominant} if for all programs $\prog' \in \fincprogs$, either $\prog' \leq_\pre \prog$ or $\prog \not\leq_\pre \prog'$.
A set $\fincprogs^\pre \subseteq \fincprogs$ of $\pre$-dominant programs is {\em complete} if for every program $\prog \in \fincprogs$, there exists a program $\prog' \in \fincprogs^\pre$ such that $\prog \leq_{\pre} \prog'$. 
Thus, for every program $\prog$ in $\fincprogs$ there is a program $\prog'$ in the set $\fincprogs^\pre$ of $\pre$-dominant programs that reaches at most those ensembles reached by $\prog$.
Dominance can be determined using the corners of polytopes.
Theorem~\ref{thm:lin_synthesis} shows that synthesis can be reduced to verification of the candidates in a finite, complete set $\fincprogs^\pre$ of dominant programs.

\begin{theorem}[Linear synthesis] Given a hardware specification $\hardware$, 
an instruction set $\synthprogs$ with instruction guard $\guard$,
a horizon $\horizon$,
a linear ensemble assertion $\pre$,
and an ensemble assertion $\post$, let $\fincprogs^\pre$ be any complete set of $\pre$-dominant  quantum programs. 
If there exists a candidate program $\prog \in \cprogs_\synthprogs$ such that (i)~$\plength{\prog}{\synthprogs} \leq \horizon$, (ii) $\prog$ is admissible w.r.t.\ $\pre$, and (iii)~$\htriple{\pre}{\prog}{\post}{\hardware}$ is valid, 
then there exists a candidate program $\prog' \in \fincprogs^\pre$ satisfying (i)-(iii).
\label{thm:lin_synthesis}
\end{theorem}

\subsection{Optimal synthesis for singular preconditions and target postconditions}
\label{subsec:optimal_singular}
When there is no uncertainty about the choice of the initial ensemble, 
we synthesize optimal programs without enumerating dominant programs.
Given a state assertion $\stateA$ and a probability value $\prob$, let
$\post(\stateA, \prob) = \probOf{\stateA} \geq \prob$ 
be the ensemble assertion that specifies the ensembles in which the hybrid states in $\sem{\stateA}$ have cumulative probability at least~$\prob$.
An ensemble assertion $\post$ is a {\em target assertion} if there exists a state assertion $\stateA$ and a probability $\prob$ such that $\post$ has the form $ \post(\stateA, \prob)$.
The set $\sem{\stateA}$ is called the {\em target set} of hybrid states.
While arbitrary postconditions require reasoning over all symbolic convex combinations of ensembles, target postconditions simplify the task because they define convex sets of ensembles.
If two ensembles each assign at least probability 
$\prob$ to the target set, then any convex combination also assigns at least 
$\prob$ to the target set.

According to the synthesis problem, we wish to find a program that reaches a target state in $\sem{\stateA}$ 
with a given probability~$\prob$.
Given a precondition $\pre$, 
a program $\prog$ is {\em optimal} for the target $\stateA$ iff 
(i) $\plength{\prog}{\synthprogs} \leq \horizon$,
(ii) $\prog$ is admissible w.r.t. $\pre$, 
(iii) there exists a probability $\prob\geq 0$ such that $\htriple{\pre}{\prog}{\post(\stateA,\prob)}{\hardware}$ is valid, 
and there is no program $\prog'$ satisfying (i) and (ii) where the triple
$\htriple{\pre}{\prog'}{\post(\stateA,\prob')}{\hardware}$ is valid for some $\prob'>\prob$. 

We use the notion of beliefs from the theory of POMDPs, and Bellman's value-iteration algorithm~\cite{bellman54} to find the optimal program for a given singular precondition and target postcondition.
A {\em quantum belief} $\belief$ is an ensemble over hybrid states such that all hybrid states in its support share the same classical component. Formally, if hybrid states \(\hstate\) and \(\hstate'\) are in \(\supp(\belief)\), then $\topure{\hstate}{\allcv} = \topure{\hstate'}{\allcv}$. 
The classical state of the quantum belief $\belief$ is denoted by $\seqobs{\belief} = \topure{\hstate}{\allcv}$ , and the set of all quantum beliefs by $\allBeliefs \subseteq \ensembles$.
For unique initial ensembles, the ensemble graph is replaced by the quantum-belief game graph, which is a decision graph and allows a richer analysis.

\begin{definition}[Quantum-belief game graph (QBGG)] For a hardware specification $\hardware$ and an instruction set $\synthprogs$ with instruction guard $\guard$, the {\em quantum-belief game graph}  $\beliefGraph(\hardware, \guard, \synthprogs)=\langle\allBeliefs,\pomdpTransF \rangle$ is a pair consisting of the set $\allBeliefs$ of quantum beliefs and a stochastic transition function $\pomdpTransF: \allBeliefs \times \synthprogs \times \allBeliefs \rightarrow [0,1]$, such that for all beliefs $\belief, \belief' \in \allBeliefs$ and instructions $\prog \in \synthprogs$, 
we have $\pomdpTransF(\belief, \prog, \belief') = \condProb{\sem{\prog}_{\hardware}(\belief)}{\mvals{\allcv} = \seqobs{\belief'}}$ if $\branchAt{\sem{\prog}_\hardware(\belief)}{\mvals{\allcv}}{\seqobs{\belief'}} = \belief'$ and $\belief \vDash \guard(\prog)$. Otherwise, $\pomdpTransF(\belief, \prog, \belief') = 0$.
\end{definition}
Intuitively, given a quantum belief $\belief$, the instruction $\prog$  has a positive probability of producing a successor belief iff $\prog$ is admissible in $\belief$, and 
the successors are the normalized beliefs produced by $\prog$. 
Quantum-belief game graphs allow us to synthesize quantum programs that go beyond straight-line programs, as they use the classical state of a belief to generate conditional blocks.
Since target postconditions are convex, 
they can be verified over beliefs without computing the overall ensemble produced by a program.
The maximum probability of reaching a target state in $\sem{\stateA}$ within $\horizon$ steps is found by the value-iteration algorithm performing $\horizon$ Bellman updates:  
{
$
    \mathcal{V}_{i+1}(\belief) = \max
        \left(\mathcal{V}_{0}(\belief),
        \max_{\prog \in \synthprogs}\left( \sum_{\belief' \in \allBeliefs}\pomdpTransF(\belief, \prog, \belief') \cdot \mathcal{V}_{i}(\belief')\right)\right)
$
}
where $0 \leq i < \horizon$ and $\mathcal{V}_{0}(\belief) = \sum_{\hstate \in \sem{\stateA}} \belief(\hstate)$. 
In the general case, if we are given an arbitrary ensemble  $\ensemble$ (not a belief), we can start value iteration by adding one extra iteration that partitions $\ensemble$ into beliefs: $\mathcal{V}_{\horizon}(\ensemble) = \sum_{\cstate \in \cstates}\condProb{\ensemble}{\mvals{\allcv} = \cstate}
\cdot \mathcal{V}_{\horizon}(\branchAt{\ensemble}{\mvals{\allcv}}{\cstate})$.
The program that we obtain from the {\em optimal policy} for an initial belief $\belief$ at horizon $0\le i \le \horizon$ is  $\prog_{\belief}^{(i)} := \halt$, if $i = 0 $ or $ 
\mathcal{V}_0(\belief) \ge 
    \sum_{\belief' \in \allBeliefs}
        \pomdpTransF(\belief, \prog^*, \belief')
        \cdot 
        \mathcal{V}_{i-1}(\belief')$. Otherwise, $\prog_{\belief}^{(i)} :=
        \prog^*;\IfElse{\mvals{\allcv} = \seqobs{\belief'_0}}{
            \prog^{(i-1)}_{\belief'_0}
        }{
        \ldots \texttt{ else \{}\prog^{(i-1)}_{\belief'_{n-1}}\texttt{\}}
        },$
where
$\prog^* =
\text{argmax}_{\prog \in \synthprogs}
\sum_{\belief' \in \allBeliefs}
\pomdpTransF(\belief, \prog, \belief')
\cdot 
\mathcal{V}_{i-1}(\belief')$ is the instruction leading to an optimal probability, and $\belief_j' \in \supp(\pomdpTransF(\belief, \prog^*))$ are the corresponding successor beliefs. 
The optimal program $\prog^{(\horizon)}_\ensemble$ for an arbitrary initial ensemble $\ensemble$ starts with a sequence of conditional blocks that partition $\ensemble$ into beliefs, and then each conditional block contains the optimal program for the corresponding belief.
Theorem~\ref{thm:bellman_correct} is shown by induction on the horizon length and can be computed exactly using arbitrary precision~\cite{Muroya25}.

\begin{theorem}[Target postcondition synthesis] 
\label{thm:bellman_correct}
Given a hardware specification $\hardware$, 
an instruction set $\synthprogs$ with instruction guard $\guard$, 
a horizon $\horizon$,
a singular ensemble assertion $\pre$, and 
a state assertion $\stateA$,
the hardware-dependent Hoare triple
$\htriple{\pre}{(\prog^{(\horizon)}_{\ensemble_\pre})}{\post(\stateA, \mathcal{V}_{\horizon}(\ensemble_\pre))}{\hardware}$ is valid, and the quantum program $\prog^{(\horizon)}_{\ensemble_\pre}$ obtained from the optimal policy  is optimal for the precondition $\pre$ and target~$\stateA$.
\end{theorem}

\begin{example}[Synthesis] 
Recall the Example~\ref{ex:motivation} in the introduction and denote by $\pre$ the precondition in Equation~\ref{eq:mot_example},
and consider an instruction set $\synthprogs = \{\instruction_{\Had}, \instruction_{\X}, \instruction_{\Z}\}$ consisting of a Hadamard $\instruction_{\Had} = \Had(\qvar)$, the bit-flip instruction $\instruction_{\X} = \X(\qvar)$, and phase-flip instruction $\instruction_{\Z} = \Z(\qvar)$.
The hardware specifications $\hardware_0$ and $\hardware_1$ define a noiseless channel  for the Hadamard instruction: $\hardware_0(\instruction_{\Had}) = \hardware_1(\instruction_{\Had}) = \{\Id\}$. However, the bit-flip instruction is noiseless on $\hardware_0$, and noisy on $\hardware_1$: i.e., $\hardware_0(\instruction_{\X}) = \hardware_1(\instruction_{\Z}) = \{\Id \}$ and $\hardware_1(\instruction_{\X}) = \hardware_0(\instruction_{\Z})$. 
At horizon $\horizon=3$ and state assertion $\stateA = (\Id = -\ket{1})$, our tool determines in less than 0.03 seconds that $(\instruction_{\Had};\instruction_{\X};\instruction_{\Had})^{\hardware_0}$ is optimal for $\hardware_0$ and $(\instruction_{\Z})^{\hardware_1}$ for $\hardware_1$.
\label{ex:synthesis}
\end{example}

\section{Probabilistic programs}
To synthesize optimal programs for linear preconditions and target postconditions, we need probabilistic statements~\cite{brice26}.

\subsection{Noise-aware Quantum Hoare Logic for probabilistic programs}
\paragraph{Probabilistic statements.} We extend the programming language of Figure~\ref{fig:pl} to include probabilistic statements as shown in Figure~\ref{fig:ppl}. We denote the set of all probabilistic programs as $\pprogs$.

\begin{figure}[t!]
    \centering
    \scriptsize
    \[
\begin{array}{rcl}
\texttt{Quantum program} \quad \prog \in \progs & &\\
\texttt{Probabilistic quantum program }\quad \pprog & := & \prog \mid \pprog_0 \oplus_{\prob} \pprog_1.
\end{array}
\]
\begingroup \scriptsize
    \[
\begin{array}{l}
    \sem{\pprog_0 \oplus_{\prob} \pprog_1}_{\hardware}(\ensemble) = (1-\prob)\cdot \sem{\pprog_0}_{\hardware}(\ensemble)
      + \prob\cdot\sem{\pprog_1}_{\hardware}(\ensemble)
\end{array}
\]\endgroup
    \caption{Extends the programming language from Figure~\ref{fig:pl} to include a probabilistic statement, where $\prob \in [0,1]$ is a probability.} 
    \label{fig:ppl}
\end{figure}

\paragraph{Hoare logic for probabilistic programs.}
Given a probabilistic term $\pterm$ the syntactic substitution for a probabilist program $\pprog =  \pprog_0 \oplus_{\prob} \pprog_1$ is given by:
{\scriptsize\begin{equation}
    \subsfunc{\hardware}{\pprog}{\pterm} = \begin{dcases*}
           (1-\prob)\cdot\subsfunc{\hardware}{\pprog_0}{\pterm} +\prob\cdot\subsfunc{\hardware}{\pprog_1}{\pterm} & $\pterm = \real\cdot\probOf{\stateA, \op}$ \\
        \subsfunc{\hardware}{\pprog}{\pterm_0} + \subsfunc{\hardware}{\pprog}{\pterm_1}& $\pterm = \pterm_0 + \pterm_1$
    \end{dcases*}
    \label{eq:prob_atom_pterm_subs}
\end{equation}}

The proof system of Figure~\ref{fig:proof_system} also extends using the following rule:
\begin{figure}[t]
    \tiny
    \[
    \inferrule*[right=Prob]{
        }{
            \htriple{\weakp{\hardware}{\pprog_0 \oplus_{\prob} \pprog_1}{\post}}{\pprog_0 \oplus_{\prob} \pprog_1}{\post}{\hardware}
        } 
    \]
    \Description{In addition to the rules of the proof system of Figure~\ref{fig:proof_system}, probabilistic programs use the following rule.}
    \caption{In addition to the rules of the proof system of Figure~\ref{fig:proof_system}, probabilistic programs use the following rule.}
    \label{fig:prob_hl_rule}
\end{figure}

\subsection{Optimal synthesis for linear preconditions and target postconditions}

The candidate programs for synthesis also consider now probabilistic statements:
$\begingroup 
    \cprog  =   \synthins \mid \cprog_0;\cprog_1 \mid  
    \IfElse{\bexpr}{\cprog_0}{\cprog_1} \mid \cprog_0 \oplus_{\prob} \cprog_1 \mid \halt
\endgroup$.

\paragraph{Dominant programs.}
We must refine the notion of dominance for target postconditions.
Consider a linear precondition $\pre$ and a target postcondition $\post(\stateA,p)$.
For two candidate programs $\prog_1,\prog_2\in\fincprogs$, 
let $\prog_1 \leq_{\pre,\stateA} \prog_2$ if for all probabilities $\prob$ such that $\{\pre\}\prog_1^\hardware\{\post(\stateA,\prob)\}$ is valid, also $\{\pre\}\prog_2^\hardware\{\post(\stateA,\prob)\}$ is valid. 
A program $\prog\in\fincprogs$ is $(\pre,\stateA)$-{\em dominant} if for all programs $\prog' \in \fincprogs$, either $\prog' \leq_{\pre,\stateA} \prog$ or $\prog \not\leq_{\pre,\stateA} \prog'$.
Let $\fincprogs^{\pre,\stateA}$ be a complete set 
containing $w$ different $(\pre,\stateA)$-dominant programs.
Since minima are attained at the corners, we can compute a
suitable set $\fincprogs^{\pre,\stateA}$ by
incrementing the horizon and evaluating only the corners of each polytope of the precondition using the algorithm from~\cite{brice26}.

\begin{lemma}Given a hardware specification $\hardware$, 
an instruction set $\synthprogs$ with instruction guard $\guard$,
a horizon $\horizon$,
a linear ensemble assertion $\pre$,
and an ensemble assertion $\post$,
if there exists a program $\prog \in \pprogs$ such that (i)~$\prog$ has length at most $\horizon$, (ii) $\prog$ is admissible w.r.t.\ $\pre$, and (iii)~
$\htriple{\pre}{\prog}{\post}{\hardware}$ is valid, 
then there exists a convex combination of deterministic programs $\prog_i \in \fincprogs$ with real coefficients $\prob_i\in [0,1]$ such that 
$\htriple{\pre}{(\bigoplus_{\prob_i}\prog_i)}{\post}{\hardware}$ is valid.
\label{lemma:cvx-det-programs}
\end{lemma}
\noindent
{Lemma}~\ref{lemma:cvx-det-programs} follows {from Kuhn's theorem}~\cite{Kuhn53} in a game-theoretical setting.
The theorem states a formal equivalence between mixed strategies (i.e., strategies that assign probabilities to deterministic strategies) and behavioral strategies (i.e., strategies that assign probabilities to actions at each decision point).
Kuhn's theorem shows that, over finite horizons, every mixed strategy has an equivalent behavioral strategy that yields the same outcome (and vice versa) whenever the adversary does not have access to the probabilities of mixed strategies.

The optimal program chooses each deterministic candidate program $\prog_i$ in $\fincprogs^{\pre,\stateA}$  
with probability~$\prob_i$. 
For a linear precondition $\pre$ and target assertion $\stateA$, we can compute the optimal values $\prob_i$ by linear programming. 
The initial polytope $\conv(\polytope)$ and ensemble in $\polytope$ are chosen adversarially. 
We find the convex coefficients 
$\mvals{x} = [x_0, \ldots, x_{w - 1}]$
for the $w$ candidate programs, and a probability vector
$\mvals{y} = [y_0, \dots, y_{N - 1}]$
over initial ensembles,
where $N = |\bigcup_ {\conv(\polytope) \in \prepolys(\pre)} \polytope|$ is the total number of corners, 
by solving a maximin problem of the form
$\max_{\mvals{x}} \min_{\mvals{y}} \; \mvals{x} \cdot M \cdot \mvals{y}^T$.
The matrix \( M \in [0,1]^{w \times N} \) stores the probabilities that we reach a target state in $\sem{\stateA}$ by executing the program $\prog_i$ from the initial ensemble $\ensemble_j \in \bigcup_{\conv(\polytope) \in\prepolys(\pre)} \polytope$,
namely, $M_{ij} = \sum_{\hstate \in \sem{\stateA}}\getFinalEnsemble(\MC(\hardware, \synthprogs), \langle\prog_i, \ensemble_j\rangle)(\hstate)$.
Since the postcondition $\post(\stateA,\prob)$ is convex and minima are attained at corners, we obtain the linear program~\cite{brice26}
{
\begin{equation}
\begin{array}{rl}
LP(\pre, \stateA) := \text{ maximize} \quad \prob  \quad
\text{subject to} & \sum\limits_i x_i = 1,
                  \quad x_i \geq 0 \quad \text{for all } 0 \leq i < w, \\
                  & \sum\limits_i x_i \cdot M_{ij} \geq \prob \quad \text{for all } 0 \leq j < N.
\end{array}
\label{eq:linear_program}
\end{equation}
}

\begin{theorem}[LP Correctness] Given a hardware specification $\hardware$, 
an instruction set $\synthprogs$ with
an instruction guard $\guard$,
a horizon $\horizon$,
a linear ensemble assertion $\pre$,
and a state assertion $\stateA$,
let $\prob$ and $\mvals{x}$ be the probability value and probability vector found by the linear program $LP(\pre, \stateA)$.
Then $\htriple{\pre}{(\bigoplus_{\mvals{x}_i}~\prog_i)}{\post}{\hardware}$ is valid.
Moreover, $\bigoplus_{\mvals{x}_i}~\prog_i$ is optimal for the precondition $\pre$ and target~$\stateA$.
\label{thm:lp_correctness}
\end{theorem}

\begin{example}[Optimal linear precondition synthesis] 
Consider the target postcondition of Example~\ref{ex:synthesis} at the end of Subsection~\ref{subsec:optimal_singular}.
Also, let $\pre$  be a linear preconditions describing a single polytope $\{\conv(\polytope)\} = \prepolys(\pre)$, in which $\polytope = \{\ensemble_0, \ensemble_1\}$ is a set containing two corners: $\ensemble_0 = \{(\ket{0}, 0): 1\}$ and $\ensemble_1 = \{(\ket{1}, 0): 1\}$.
Optimal synthesis for linear preconditions at horizon $\horizon=2$, using the instruction set from  Example~\ref{ex:synthesis} at the end of Subsection~\ref{subsec:optimal_singular}, using a noiseless hardware gives the optimal program $\htriple{\pre}{((\instruction_{\X};\instruction_{\Z}) \oplus_{0.5} (\instruction_{\Z}))}{\post(\stateA, 0.5)}{\hardware}$.
Thus, the optimal program mixes two deterministic programs with equal probabilities to mitigate their extreme values.
In contrast, a deterministic program can only be optimal for the target $\stateA$ with probability 0.
\end{example}

\begin{table}[t]
\scriptsize
	\begin{center}
		\begin{tabular}{c|c|c|c|c|c|c|c|c|c|c|c}
			problem & \makecell[c]{instr.\\ set ($\synthprogs$)} & prec. type & \makecell[c]{init.\\ states ($\initSet$)} & \makecell[c]{target\\ states ($\targetSet$)} & \#emb. & \makecell[c]{$\horizon_{min}$} & \makecell[c]{$\horizon_{max}$} & \makecell[c]{\#diff.\\progs.} & \makecell[c]{\#tot.\\ improv.} & \makecell[c]{avg. synth.\\ time (s)} & \makecell[c]{avg. verif.\\ time (s)} \\
			\hline
			\multirow{3}{*}{\IPMAPAbbr} & \multirow{2}{*}{\IPMAPIns} & singular & \multirow{3}{*}{\makecell[c]{$\ket{\bstateS^{\pm}0}$, \\$\ket{\bstateSX^{\pm}0}$}} & \multirow{3}{*}{$\ket{\bstateS^{\pm}}_{[\qvar_0, \qvar_1]}$} & 141 & 3 & 8 & 34 & 32 & 580.71 & 1.48 \\
			\cline{3-3} \cline{6-12}
			     &  & linear & & & 141 & 3 & 5 & 16 & 188 & 214.16 & 0.82 \\
			\cline{2-3} \cline{6-12}
			      & \CXHIns & singular &  &  & 203 & 7 & 7 & 7 & 24 & 2856.12 & 0.26 \\
			\hline
			\GHZAbbr & \GHZIns & singular & $\ket{000}$ & $\ket{\text{GHZ}}$ & 1193 & 3 & 3 & 12 & - & 13.58 & 0.14 \\
			\hline
			\multirow{2}{*}{\ResetAbbr} & \multirow{2}{*}{\ResetIns} & singular & \multirow{2}{*}{$\ket{0},\ket{1}$} & \multirow{2}{*}{$\ket{0}$} & 109 & 2 & 9 & 35 & 458 & 23.9 & 1.5 \\
			\cline{3-3} \cline{6-12}
			 &  & linear & & & 109 & 2 & 7 & 31 & 415 & 31.4 & 0.05 \\
			\hline
			\multirow{2}{*}{\LPhaseAbbr} & \multirow{2}{*}{\LPhaseIns} & singular & \multirow{2}{*}{\makecell[l]{$\ket{{+}0},$\\ $\ket{{+}1}$}} &\multirow{2}{*}{ \makecell[l]{$\ket{{-}0},$\\$-\ket{{-}1}$}} & 184 & 2 & 7 & 16 & 553 & 5.49 & 0.06 \\
			\cline{3-3} \cline{6-12}
			 &  & linear & & & 184 & 2 & 7 & 49 & 917 & 6.71 & 0.07
		\end{tabular}
	\end{center}
	\caption{For the \Paritybitflip~problem, we performed experiments with two different instruction sets, $\IPMAIns$ and $\CXHIns$. Each row fixes a problem, an instruction set ($\synthprogs$), and a precondition type: the singular precondition is the uniform mixture of the initial states in $\initSet$, and the linear precondition is their polytope. The postcondition is the target assertion for the states in $\targetSet$. We also report the number of considered logical-to-physical qubit embeddings, the minimum and maximum  horizons used to synthesize programs, the number of distinct programs synthesized across all considered embeddings and  horizons, the number of improvements relative to the baseline textbook program, and the average time (in seconds) required to synthesize and verify these programs.}
	\label{tab:exp_summary}
\end{table}

\section{Experiments}
\label{sec:experiments}
We ran experiments using 16GB of RAM and an Intel(R) Xeon(R) CPU E5-2680 v3 @2.50GHz processor.
We built a simulation, verification, and bounded-synthesis toolchain in C++17 and applied it to four problems (Table~\ref{tab:exp_summary})  drawn from the literature on all 55 hardware specifications available on Qiskit v2.2.1~\cite{Qiskit}.
Each problem uses at most one classical variable~$\cvar$ and   an instruction set from  Equation~\ref{eq:instruction_sets}.
They include all experiments that were performed in~\cite{Muroya25} with  a previous version of Qiskit.
We are generally able to explore larger horizons over~\cite{Muroya25}.
All numerical results are rounded to 3 decimal places for readability.

Each problem is given by a set $\initSet$ of initial quantum states and a set $\targetSet$ of target states.
Writing $\ensemble_{\qstate}$ for the ensemble that assigns probability~1 to $(\ket{\qstate},0)$, we use two kinds of preconditions:
the \emph{singular} precondition $\pre_{\bar\ensemble}$, where $\bar\ensemble$ assigns equal probability to every $(\ket{\qstate},0)$ with $\ket{\qstate}\in\initSet$, and
the \emph{linear} precondition $\pre_\polytope$ with $\polytope = \{\ensemble_{\qstate}\mid\ket{\qstate}\in\initSet\}$.
The postcondition is the target assertion $\post(\stateA_\targetSet,\prob)$ with $\stateA_\targetSet = \bigvee_{\ket{\qstate}\in\targetSet}(\Id = \ket{\qstate})$.
For linear preconditions, the probability $\prob$ of a $\post(\stateA_\targetSet,\prob)$ can be intuitively interpreted as the probability that the worst-case scenario (i.e., initial ensemble) reaches the target set.
Our experiments report the optimal programs, which for  linear preconditions correspond to the program maximizing the worst-case scenario.
For each embedding (mapping of logical into physical qubits), we synthesize optimal programs for every horizon $\horizon_{min}\leq\horizon\leq\horizon_{max}$ and verify each synthesized program against $\post(\stateA_\targetSet,\prob^*)$, where $\prob^*$ is its computed success probability.
The number of distinct   programs across all horizons and embeddings (Table~\ref{tab:exp_summary})  shows that optimality is hardware-dependent.
The synthesized programs discussed in this section are all in the Appendix, Section~\ref{appendix:programs}.

Denote by  $\ket{\bstateSX^{\pm}} = (\ket{01}\pm\ket{10})/\sqrt{2}$ the odd-parity Bell states. 
The \Paritybitflip~problem over three quantum variables $\qvars=[\qvar_0, \qvar_1, \qvar_2]$ consists of optimizing the probability to reach even-parity Bell states $\ket{\bstateS^{\pm}}$  (defined in Example~\ref{ex:motivation}) while preserving the Bell state. So we put a guard on measurement instructions.
This is a common task in quantum error correction~\cite{qc_bible,gottesman97}, achieved by performing indirect parity measurements (IPM).
The baseline for the $\IPMAIns$ instruction set is the textbook program~\cite{qc_bible} with majority-vote strategy (see Appendix Fig.~\ref{fig:alg_perfect_ipma2}), that performs measurements  and applies a bit-flip instruction to change the parity of a Bell state depending on the majority of measurement outcomes.
The baseline for the $\CXHIns$ instruction set is the optimal program synthesized for the noiseless hardware (see Appendix,  Fig.~\ref{fig:cxh_programs}(a)).

The \GHZsp~problem consists of reaching with high probability the GHZ state $\ket{\text{GHZ}} = (\ket{000} + \ket{111})/\sqrt{2}$~\cite{Greenberger89}. 
The task of preparing a GHZ state is relevant since the GHZ state has several applications~\cite{Eldredge18,qc_bible,Hillery98} in protocols and quantum error correction.
For this problem, we  synthesize the optimal program for every possible embedding of three qubits in every quantum hardware. 
Due to hardware connectivity,  not all instructions in the $\GHZIns$ instruction set are always available.
On Yorktown, the only quantum hardware for which it is possible to test all 12 programs, different programs achieve success probabilities  between $93.596\%$ and $96.376\%$ (see Appendix, Fig.~\ref{fig:ghz_experiment}) in the same embedding.

The \StateReset~problem consists of resetting a qubit to state $\ket{0}$ and is a basic task required by quantum programs in general~\cite{Zhou21R,qc_bible,Geher24}.
We also use the textbook algorithm with a majority-vote  strategy as baseline (see Appendix, Fig.~\ref{fig:reset_algs}(a)), conceptually similar to the one for the \Paritybitflip~problem.

The \PhasePrep~problem extends  Example~\ref{ex:motivation} to consider a scenario in which measurement outcomes are transmitted using a classical variable and those outcomes  determine a particular phase preparation procedure.
Our problem is not expressible under density-matrix semantics, yet is exactly what protocols such as quantum teleportation~\cite{Bennett93,qc_bible} perform upon measurement outcomes sent by a classical channel.
This problem uses a classical variable $\cvar$ and 2 quantum variables $\qvars = [\qvar_0, \qvar_1]$.
On noiseless hardware,  measurements reveal which out of two corrections to perform: (1)~$\Z(\qvar_0)$, or (2)~$\Z(\qvar_0);\X(\qvar_0)$.
The baseline for horizon $\horizon\leq 2$ is the phase-flip instruction $\Z(\qvar_0)$. For greater horizons, the baseline uses majority-vote (see Appendix, Fig.~\ref{fig:lphase-baseline}).

\noindent
{\em Improvements with respect to baseline.} 
Table~\ref{tab:ex_improv} shows the improvements that are achieved by each instruction set with respect to the baseline comparison. 
Since the maximum improvement naturally occurs on the most noisy embeddings, we also report an example of an improvement for a program that reaches the target set with probability $\prob^* \geq 0.8$.

Optimal programs often do not maximize the number of measurements; instead, they interleave unitary instructions and measurements in hardware-specific patterns, balancing the information gained from each measurement against noise.
Linear preconditions lead to more distinct optimal programs and larger improvements than singular ones, showing that hardware-dependent synthesis matters even more when the initial state is unknown.
Also, noise can lead to higher worst-case probabilities.
For example, for the \PhasePrep{} problem at horizon~2, both $\Z(\qvar_0)$ and $\Z(\qvar_0);\X(\qvar_0)$ have worst-case success probability~0 on noiseless hardware.
Under noise, however, the sequence $\Z(\qvar_0);\X(\qvar_0)$  has in general more noise than  $\Z(\qvar_0)$ alone, so the former reaches the target in the worst-case initial ensemble with a higher non-zero probability compared to the  phase-flip instruction alone. So it is optimal on all hardware specifications.

\begin{table}[t]
\small
	\begin{center}
		\begin{tabular}{|c|c|c|c|c||c|c|c|c|}\hline 
            & \multicolumn{4}{c||}{maximum improvement} & \multicolumn{4}{c|}{improvement example ($\prob^*\geq 0.8$)}\\
			\hline
			 & baseline & optimal & diff. & Fig. & baseline & optimal & diff. & Fig. \\
			\hline
			\IPMAPIns~(singular) & 47.575\% & 52.425\% & 4.85\% & \ref{fig:alg_optimal_ipma2} & 93.101\% & 94.0\% & 0.899\% & \ref{fig:alg_optimal_ipma2_cambridge} \\
			\hline
			\IPMAPIns~(linear) & 37.215\% & 58.583\% & 21.368\% & \ref{fig:ipma_lin_programs} & 72.655\% & 84.254\% & 11.599\% & \ref{fig:ipma_lin_programs} \\
			\hline
			\CXHIns~(singular) & 47.007\% & 50.2\% & 3.193\% & \ref{fig:cxh_programs}(b) & 89.868\% & 91.438\% & 1.57\% & \ref{fig:cxh_programs}(c) \\
			\hline
			\ResetIns~(singular) & 36.224\% & 63.878\% & 27.654\% & - & 87.69\% & 91.347\% & 3.657\% & - \\
			\hline
			\ResetIns~(linear) & 22.673\% & 61.948\% & 39.275\% & \ref{fig:alg_optimal_reset} & 81.576\% & 90.309\% & 8.733\% & \ref{fig:reset_algs}(b) \\
			\hline
			\LPhaseIns~(singular) & 40.696\% & 62.072\% & 21.376\% & \ref{fig:alg_optimal_lphase} & 83.374\% & 90.709\% & 7.335\% & \ref{fig:alg_robust_lphase} \\
			\hline
			\LPhaseIns~(linear) & 17.504\% & 59.699\% & 42.195\% & \ref{fig:alg_optimal_lphase_lin} & 74.079\% & 91.637\% & 17.558\% & \ref{fig:alg_robust_lphase_lin} \\
			\hline
		\end{tabular}
	\end{center}
	\caption{For each instruction set, we consider two types of improvements: the maximum improvement over the baseline program across all considered horizons and embeddings, and an example improvement achieved in a more robust embedding. For both, we report the success probabilities of the baseline and optimal programs, their difference, and references to the corresponding figures in Appendix Section~\ref{appendix:programs}. The \StateReset~problem does not include some figure references because those programs exceed 1000 lines.}
	\label{tab:ex_improv}
\end{table}

\section{Conclusion}
We presented a noise-aware quantum Hoare logic (\nQHL) over ensembles with a relatively complete proof system that enables  
modular reasoning without separability assumptions.
We demonstrated how the logic can be used for
numerically precise model checking of loop-free quantum programs (a.k.a.\ bounded model checking) on noisy hardware, and for automatically 
synthesizing loop-free quantum programs for common tasks, tailored to a given  hardware specification. 
Synthesizing hardware-optimal programs for arbitrary  preconditions remains an open challenge.

\bibliographystyle{splncs04}
\bibliography{references}

\newpage
\spnewtheorem*{lemma*}{Lemma}{\bfseries}{\itshape}
\spnewtheorem*{theorem*}{Theorem}{\bfseries}{\itshape}

\section{APPENDIX: Lemma~\ref{lemma:dual_pterm}}
\label{appendix:dual_pterm}
\begin{lemma*}[Probabilistic substitutions] 
For all probabilistic terms $\pterm$, ensembles $\ensemble \in \ensembles$, quantum programs $\prog \in \progs$, and quantum hardware $\hardware$; 
if $\ensemble' = \sem{\prog}_\hardware(\ensemble)$ then $\sem{\subsfunc{\hardware}{\prog}{\pterm}}(\ensemble) = \sem{\pterm}(\ensemble')$.
\end{lemma*}
\begin{proof}
    Consider a quantum program $\prog \in \progs$, a quantum hardware $\hardware$, an ensemble $\ensemble \in \ensembles$ and a state assertion $\stateA$.
    To prove Lemma~\ref{lemma:dual_pterm} it is enough to show that it holds for probabilistic terms of the form $\pterm = \probOf{\stateA, \qterm}$. 
    Let $\ensemble' = \sem{\prog}_\hardware(\ensemble)$ be the successor ensemble.
    The proof is by induction on the program $\prog$.
    \begin{itemize}
    \item If $\prog = \Skip$ then $\ensemble = \ensemble'$ therefore $\sem{\subsfunc{\hardware}{\Skip}{\pterm}}(\ensemble) = \sem{\pterm}(\ensemble) = \sem{\pterm}(\ensemble')$.
    \item If $\prog = (\cvar = \bterm)$, we have {\small $$
    \sem{\pterm}(\ensemble') = \sum_{\hstate \in \sem{\csubs{\stateA}{\cvar}{\bterm}}}\ensemble'(\hstate)\cdot 
    \norm{\sem{\qterm}(\hstate)}^2 = \sum_{\hstate \in \sem{\csubs{\stateA}{\cvar}{\bterm}}} \ensemble(\hstate)\cdot 
    \norm{\sem{\qterm}(\hstate)}^2
    = \sem{\csubs{\pterm}{\cvar}{\bterm}}(\ensemble).
    $$}
    \item If $\prog = (\cvar = \Measure(\qvar))$ then $\ensemble' = \ensemble_0 + \ensemble_1$ can be written in terms of two subensembles $\ensemble_{\bit} = \ensembleAt{\ensemble'}{\cvar = \bit}$ where $\bit \in \{0,1\}$.
    Let $\schannel = \hardware(\langle\qvar, \cvar \rangle)$ be the corresponding measurement channel.
 Since the real measurement outcome $\bit$ depends on the quantum hardware, we have 
    {$$
   \sem{\pterm} (\ensemble_{\bit}) = \sum_{\bit' \in \{0,1\}}\schannel(\bit, \bit')\cdot
    \left(\sum_{\branchAt{\hstate}{}{\projector_{\bit'}} \in \sem{\csubs{\stateA}{\cvar}{\bit}}}\ensemble(\hstate)\cdot \norm{\sem{\qsubs{\qterm}{}{\projector_{\bit'}}}(\hstate)}^2 \right).
    $$}
    Therefore, 
    {\scriptsize\begin{equation*}
    \begin{aligned}
        \sem{\pterm}(\ensemble') = \sum_{\bit\in \{0,1\}}\sem{\pterm}(\ensemble_\bit) &=\sum_{\bit,\bit' \in \{0,1\}}\schannel(\bit, \bit')\cdot
    \left(\sum_{\hstate \in \sem{\qsubs{\csubs{\stateA}{\cvar}{\bit}}{[\qvar]}{\projector_{\bit'}}}}\ensemble(\hstate)\cdot  \norm{\sem{\qsubs{\qterm}{[\qvar]}{\projector_{\bit'}}}(\hstate)}^2 \right)
    \\
        & = \sum_{\bit,\bit' \in \{0,1\}}\schannel(\bit, \bit') \cdot\left(\sem{\qsubs{\csubs{\pterm}{\cvar}{\bit}}{[\qvar]}{\projector_{\bit'}}}(\ensemble)\right) = \sem{\subsfunc{\hardware}{\prog}{\pterm}}(\ensemble).
    \end{aligned}
    \end{equation*}}
    \item If $\prog = U(\qvars)$ we have 
    {\small \begin{equation*}
    \begin{aligned}
        \sem{\pterm}(\ensemble') 
        &= \sum_{\op \in \hardware(\langle U, \qvars \rangle)} \left( \sum_{
            \hstate \in \sem{\qsubs{\stateA}{\mvals{\allqv}}{\op\cdot U_{\qvars}}}
        } \ensemble(\hstate)\cdot  \norm{\sem{\qsubs{\qterm}{\mvals{\allqv}}{\op\cdot U_{\qvars}}}(\hstate)}^2\right) \\
        &=  \sum_{\op \in \hardware(\langle U, \qvars \rangle)}
        \sem{\qsubs{\pterm}{\mvals{\allqv}}{\op\cdot U_{\qvars}}}(\ensemble).
    \end{aligned}
    \end{equation*}}
    \item If $\prog = \IfElse{\bexpr}{\prog_0}{\prog_1}$ 
    then 
        $\sem{\pterm}(\ensemble') = 
        \sem{\pterm} (\sem{\prog_0}_{\hardware}(\ensembleAt{\ensemble}{\bexpr})) + \sem{\pterm}(\sem{\prog_1}_{\hardware}(\ensembleAt{\ensemble}{\neg\bexpr})).$
    By induction hypothesis we have:
    \begin{equation*}
        \begin{aligned}
            \sem{\pterm}(\ensemble') &= 
            \sem{\pterm} (\sem{\prog_0}_{\hardware}(\ensembleAt{\ensemble}{\bexpr})) + \sem{\pterm}(\sem{\prog_1}_{\hardware}(\ensembleAt{\ensemble}{\neg\bexpr})) \\
            &= \sem{
            \rcond{\subsfunc{\hardware}{\prog_0}{\pterm}}{\bexpr}}(\ensemble) + \sem{\rcond{\subsfunc{\hardware}{\prog_1}{\pterm}}{\neg\bexpr}}(\ensemble).
        \end{aligned}
    \end{equation*}
    \item If $\prog = \prog_0;\prog_1$,
    let $\pterm' = \subsfunc{\hardware}{\prog_1}{\pterm}$.
    By induction hypothesis we have
    \begin{equation*}
        \begin{aligned}
            \sem{\subsfunc{\hardware}{\prog_0}{\pterm'}}(\ensemble) &= \sem{\pterm'}(\sem{\prog_0}_\hardware(\ensemble))
            = 
    \sem{\subsfunc{\hardware}{\prog_1}{\pterm}}(\sem{\prog_0}_{\hardware}(\ensemble)) \\
    &= \sem{\pterm}(\sem{\prog_0;\prog_1}_{\hardware}(\ensemble)) = \sem{\pterm}(\ensemble').
        \end{aligned}
    \end{equation*}
    \item If $\prog = \MyWhile{\bexpr}{\prog_0}$.
Let $\gamma_0 = \ensemble$ and $\gamma_{k+1} = \sem{\prog_0}_\hardware(\ensembleAt{\gamma_k}{\bexpr})$.
We first show, by induction on $k$, that for \emph{all} ensembles $\ensemble$,
holds 
$\sem{\Phi_\hardware^k(\rcond{\pterm}{\neg\bexpr})}(\ensemble) = \sem{\pterm}(\ensembleAt{\gamma_k}{\neg\bexpr}).$
For $k=0$ clearly holds. 
For $k+1$:
\begin{align*}
\sem{\Phi_\hardware^{k+1}(\rcond{\pterm}{\neg\bexpr})}(\ensemble)
&= \sem{\rcond{\subsfunc{\hardware}{\prog_0}{\Phi_\hardware^{k}(\rcond{\pterm}{\neg\bexpr})}}{\bexpr}}(\ensemble) \\
&= \sem{\subsfunc{\hardware}{\prog_0}{\Phi_\hardware^{k}(\rcond{\pterm}{\neg\bexpr})}}(\ensembleAt{\ensemble}{\bexpr}) \\
&= \sem{\Phi_\hardware^{k}(\rcond{\pterm}{\neg\bexpr})}\big(\sem{\prog_0}_\hardware(\ensembleAt{\ensemble}{\bexpr})\big) \\
&= \sem{\pterm}(\ensembleAt{\gamma_{k+1}}{\neg\bexpr}). 
\end{align*}
Summing over $k$,
\[
\sem{\subsfunc{\hardware}{\prog}{\pterm}}(\ensemble)
= \sum_{k=0}^{\infty} \sem{\pterm}(\ensembleAt{\gamma_k}{\neg\bexpr})
= \sem{\pterm}\Big(\sum_{k=0}^{\infty} \ensembleAt{\gamma_k}{\neg\bexpr}\Big)
= \sem{\pterm}(\sem{\prog}_\hardware(\ensemble)).
\]
    \end{itemize}
\end{proof}

\section{APPENDIX: weakest-precondition function}
\label{appendix:lemma:wpf}
\begin{lemma}[Weakest-precondition function] For all hardware specifications $\hardware$, quantum programs $\prog$, and ensemble assertions $\post$, we have $\weakp{\hardware}{\prog}{\post}$ is the weakest liberal precondition, i.e.,:
\begin{itemize}
    \item \textbf{(Soundness)} $\htriple{\weakp{\hardware}{\prog}{\post}}{\prog}{\post}{\hardware}$ is valid.
    \item \textbf{(Weakness)} for all ensemble assertions $\pre$, if $\htriple{\pre}{\prog}{\post}{\hardware}$ is valid then $\pre \Rightarrow \weakp{\hardware}{\prog}{\post}$.
\end{itemize}
\label{lemma:wpf}
\end{lemma}
\begin{proof}
    Consider some arbitrary quantum hardware $\hardware$, quantum program $\prog$, and  postcondition $\post$.
    By the termination assumption, $\sem{\prog}_\hardware(\ensemble)\in\ensembles$ for every $\ensemble\in\ensembles$.
We show by induction on $\post$ that for all ensembles $\ensemble$,
\begin{equation}
\ensemble\vDash\weakp{\hardware}{\prog}{\post} \iff \sem{\prog}_\hardware(\ensemble)\vDash\post.
\label{eq:wpf_helper}
\end{equation}
\begin{itemize}
\item Case $\post = \pterm_0\geq\pterm_1$. By Lemma~\ref{lemma:dual_pterm},
$\sem{\subsfunc{\hardware}{\prog}{\pterm_j}}(\ensemble) = \sem{\pterm_j}(\sem{\prog}_\hardware(\ensemble))$ for $j\in\{0,1\}$, so both sides of Equation~\ref{eq:wpf_helper} state the same inequality.
\item Case $\post = \neg\post'$. Since $\sem{\prog}_\hardware(\ensemble)\in\ensembles$, we have 
$\ensemble\vDash\neg\weakp{\hardware}{\prog}{\post'}$
iff $\ensemble\not\vDash\weakp{\hardware}{\prog}{\post'}$
iff $\sem{\prog}_\hardware(\ensemble)\not\vDash\post'$ 
iff $\sem{\prog}_\hardware(\ensemble)\vDash\neg\post'$.
\item Case $\post = \post_0\land\post_1$. Immediate from the induction hypotheses for $\post_0$ and $\post_1$.
\end{itemize}
\textbf{Soundness} follows from the ($\Rightarrow$) direction of Equation~\ref{eq:wpf_helper}.
\textbf{Weakness:} if $\htriple{\pre}{\prog}{\post}{\hardware}$ is valid and $\ensemble\vDash\pre$, then $\sem{\prog}_\hardware(\ensemble)\vDash\post$, so $\ensemble\vDash\weakp{\hardware}{\prog}{\post}$ by the ($\Leftarrow$) direction of Equation~\ref{eq:wpf_helper}. Hence $\pre\Rightarrow\weakp{\hardware}{\prog}{\post}$.
\end{proof}

\section{APPENDIX: Theorem~\ref{thm:proof_system}}
\label{appendix:thm:proof_system}
\begin{theorem*}[Proof system] Let $\vdash$ denote syntactic derivability using the proof system of Figure~\ref{fig:proof_system}. For all quantum programs $\prog \in \progs$, hardware specifications $\hardware$, and ensemble assertions $\pre, \post$, the following hold:
\begin{enumerate}
    \item \textbf{Soundness.} If $\vdash \htriple{\pre}{\prog}{\post}{\hardware}$,  then $\htriple{\pre}{\prog}{\post}{\hardware}$ is valid.
    \item \textbf{Weakest precondition derivability.} If $\pre$ is the weakest precondition for $\prog$ and $\post$, then 
    $\vdash \htriple{\pre}{\prog}{\post}{\hardware}$.
    \item \textbf{Relative completeness.} If $\htriple{\pre}{\prog}{\post}{\hardware}$ is valid, then $\vdash\htriple{\pre}{\prog}{\post}{\hardware}$ relative to an oracle deciding $\Rightarrow$ (entailment of assertions).
\end{enumerate}
\end{theorem*}
\begin{proof}
    \begin{enumerate}
        \item \textbf{Soundness.}  For atomic programs, and control flow statements the proof system explicitly use the weakest-precondition function whose soundness is shown in the Appendix Section~\ref{appendix:lemma:wpf},  Lemma~\ref{lemma:wpf}. Soundness for the remaining proof rules is shown as follows.
        \begin{itemize}
            \item Rule $\textsc{Seq}$. Suppose that $\htriple{\pre_0}{\prog_0}{\post_0}{\hardware}$ and $\htriple{\post_0}{\prog_1}{\post_1}{\hardware}$ are valid and let $\ensemble$ be an ensemble such that $\ensemble \vDash \pre_0$.
            Then $\sem{\prog_0}_{\hardware}(\ensemble)\vDash \post_0$ and $\sem{\prog_1}_\hardware(\sem{\prog_0}_\hardware(\ensemble)) = \sem{\prog_0; \prog_1}_\hardware(\ensemble) \vDash \post_1$.
            \item Rule $\textsc{Conseq}$. Suppose that $\htriple{\pre'}{\prog}{\post'}{\hardware}$ is valid  and we are given ensemble assertions $\pre, \post$ such that $\pre \Rightarrow \pre'$ and $\post' \Rightarrow \post$.
            For an ensemble $\ensemble$ such that $\ensemble \vDash \pre$, since $\pre \Rightarrow \pre'$, follows $\ensemble \vDash \pre'$ and  $\sem{\prog}_{\hardware}(\ensemble) \vDash \post'$.
            Since $\post' \Rightarrow \post$, then $\sem{\prog}_{\hardware}(\ensemble) \vDash\post$.
        \end{itemize}
        \item \textbf{Weakest precondition derivability.} The proof is by induction on the program syntax. For atomic programs, and control flow statements the proof system uses the function $\weakp{\hardware}{\prog}{\post}$, and weakness is shown in the Appendix Section~\ref{appendix:lemma:wpf},  Lemma~\ref{lemma:wpf}.
        If $\prog = \prog_0; \prog_1$, assume that 
        the weakest liberal precondition, for 
        $\post$ and $\prog_1$ is  derivable, i.e.,
        $\vdash \htriple{\post_0}{\prog_1}{\post}{\hardware}$. Also, suppose we can derive the weakest precondition $\pre_0$ for $\post_0$ and $\prog_0$: $\vdash \htriple{\pre_0}{\prog_0}{\post_0}{\hardware}$.
        Then the rule $\textsc{Seq}$ derives the weakest precondition for $\prog = \prog_0; \prog_1$.
        \item \textbf{Relative completeness.} Assume that $\htriple{\pre}{\prog}{\post}{\hardware}$ is valid.
        Since the weakest precondition $\pre'$ is derivable, the rule $\textsc{Conseq}$ derives $\vdash \htriple{\pre}{\prog}{\post}{\hardware}$, provided an oracle that decides $\pre \Rightarrow \pre'$.
    \end{enumerate}
\end{proof}

\section{APPENDIX: Lemma~\ref{lemma:final-ensemble}}
\label{appendix:proof-final-ensemble}
\begin{lemma*}[Final ensemble] For all loop-free quantum programs $\prog$, 
hardware specifications $\hardware$, and ensembles $\ensemble$;  
 $\getFinalEnsemble(EG(\hardware, \atomprogs), \langle\prog, \ensemble \rangle) = \sem{\prog}_\hardware(\ensemble, \langle\prog, \ensemble \rangle)$. 
 Furthermore, if $\ensemble$ is a representable ensemble, then  $\getFinalEnsemble(EG(\hardware, \atomprogs))$ can be computed exactly using $\mathcal{O}(\plength{\prog}{\atomprogs})$ precision units (e.g., bits).
\end{lemma*}
\begin{proof}
    Let $\prog \in \finprogs$ be a hybrid quantum-classical program, 
$\hardware$ a quantum hardware specification,
$\MC$ the corresponding ensemble graph, and
$\pre$ be a singular ensemble assertion.
Then 
\begin{itemize}
    \item Case $\prog \in \atomprogs$. We have 
    \begin{equation*}
        \getFinalEnsemble(\MC, \langle\prog, \ensemble\rangle) = \getFinalEnsemble(\MC, \langle\halt, \sem{\prog}_\hardware(\ensemble)\rangle) = \sem{\prog}_\hardware(\ensemble).
    \end{equation*}
    \item Case $\prog = \prog_0;\prog_1$. We assume that, for all ensembles $\ensemble \in \ensembles$, we have 
	$$\getFinalEnsemble(\MC, \langle\prog_0, \ensemble\rangle) = \sem{\prog_0}_\hardware(\ensemble) \quad \text{ and } \quad \getFinalEnsemble(\MC, \langle\prog_1, \ensemble\rangle) = \sem{\prog_1}_\hardware(\ensemble).$$
    Also, suppose that $\sem{\prog_0}_\hardware(\ensemble) = \ensemble'$ leads to $n$ different nodes in the ensemble graph. 
    Then following the first operational rule for sequences of programs we get
    \begin{equation*}
        \begin{aligned}
            \getFinalEnsemble(\MC,\langle \prog, \ensemble\rangle) = \sum _{ 0 \leq i < n} \prob_i \cdot \getFinalEnsemble(\MC, \langle \prog_1 ,\ensemble'_i\rangle)
        \end{aligned}
    \end{equation*}
    where $\ensemble' = \sum_{0 \leq i < n}\prob_i \cdot \ensemble'_i$ for some $\prob_i \in [0,1]$.
    By induction hypothesis and linearity of program semantics we have
    \begin{equation*}
        \begin{aligned}
            \getFinalEnsemble(\MC, \langle \prog, \ensemble\rangle) = \sum _{ 0 \leq i < n} \prob_i \cdot \sem{\prog_1}_\hardware(\ensemble'_i) =
            \sem{\prog_1}_\hardware(\ensemble') = \sem{\prog_1}_\hardware(\sem{\prog_0}_\hardware(\ensemble)).
        \end{aligned}
    \end{equation*}
    \item Case $\prog = \IfElse{\bexpr}{\prog_0}{\prog_1}$.  Assume that for all ensembles $\ensemble\in \ensembles$ we have  $\getFinalEnsemble(\MC, \langle\prog_0, \ensemble\rangle) = \sem{\prog_0}_\hardware(\ensemble)$ and $\getFinalEnsemble(\MC, \langle\prog_1, \ensemble\rangle) = \sem{\prog_1}_\hardware(\ensemble)$.
    By induction hypothesis:
    \begin{equation*}
    \begin{aligned}
        \getFinalEnsemble(\MC, \langle\prog, \ensemble\rangle) = & \condProb{\ensemble}{\bexpr} \cdot\getFinalEnsemble(\MC, \langle\prog_0, \branchAt{\ensemble}{}{\bexpr}\rangle) \\
        &+ \condProb{\ensemble}{\neg\bexpr}\cdot \getFinalEnsemble(\MC, \langle\prog_1, \branchAt{\ensemble}{}{\neg\bexpr}\rangle)\\
         =&\condProb{\ensemble}{\bexpr} \cdot\sem{\prog_0}_\hardware(\branchAt{\ensemble}{}{\bexpr}) + \condProb{\ensemble}{\neg\bexpr}\cdot \sem{\prog_1}_\hardware(\branchAt{\ensemble}{}{\neg\bexpr})\\
		 =&\sem{\IfElse{\bexpr}{\prog_0}{\prog_1}}_{\hardware}(\ensemble).
    \end{aligned}
    \end{equation*}
\end{itemize}
Each iteration in the recursive computation of $\getFinalEnsemble$ increases linearly the precision required for the denominators of all rational numbers needed for representing ensembles, and the number of iterations is bounded by depth $\horizon = \plength{\prog}{\atomprogs}$.
More concretely, a program of depth $\horizon$ performs at most $\horizon$ multiplications of rational numbers that are multiples of $1/\ell$ for some integer $\ell > 0$. 
A multiplication between such probabilities yields a new probability multiple of $1/\ell^2$. By induction, $\horizon$ multiplications yield a probability that is multiple of $1/\ell^\horizon$. 
So if we need $\precision$ bits of precision to represent $\ell$, i.e., $\precision = \lceil\log_2(\ell)\rceil$; the final precision required is $\lceil\log_2(\ell^\horizon)\rceil \leq \horizon \cdot \precision$.
Hence, if all probabilities for representing $\ensemble$ and the transition probabilities are integer multiples of $1/\ell$, the final ensemble can be represented by probabilities that are multiples of $1/\ell^{\horizon+1}$~\cite{Muroya25}; and  $\getFinalEnsemble$ can be computed exactly by storing numbers with precision 
$(\horizon+1)\cdot \precision$.
\end{proof}

\section{APPENDIX: Lemma~\ref{lemma:real-formula}}
\label{sec:proof-real-formula}
\begin{lemma*}[First-order formula] For all representable ensembles $\ensemble$ and all ensemble assertions $\post$, the first-order formula $\textsc{realForm}(\ensemble, \post)$
is satisfiable in the theory of real-closed fields iff $\ensemble\vDash \post$.
\end{lemma*}
\begin{proof}
We first show by induction on $\pterm$ that for every representable ensemble $\ensemble$, the term $\textsc{realTerm}(\ensemble,\pterm)$ denotes the real number $\sem{\pterm}(\ensemble)$.
For $\pterm = \real\cdot\probOf{\stateA,\op}$, since $\ensemble(\hstate) = 0$ for $\hstate\notin\supp(\ensemble)$,
\[
\sem{\pterm}(\ensemble) = \real\cdot\sum_{\hstate\in\sem{\stateA}}\ensemble(\hstate)\cdot\norm{\sem{\op}(\hstate)}^2
= \sum_{\hstate\in S'}\real\cdot\ensemble(\hstate)\cdot\norm{\sem{\op}(\hstate)}^2
= \textsc{realTerm}(\ensemble,\pterm),
\]
where $S' = \{\hstate\in\supp(\ensemble)\mid\hstate\in\sem{\stateA}\}$.
The case $\pterm = \pterm_0+\pterm_1$ follows from the induction hypothesis.

We then show by induction on $\post$ that $\textsc{realForm}(\ensemble,\post)$ holds iff $\ensemble\vDash\post$.
For $\post = \pterm_0\geq\pterm_1$, this follows from the claim above.
The cases $\neg\post'$ and $\post_0\land\post_1$ follow from the induction hypothesis, since $\textsc{realForm}$ commutes with $\neg$ and $\land$ and $\sem{\cdot}$ interprets them as complement and intersection.
\end{proof}

\section{Appendix: Lemma~\ref{lemma:lin_prec_formula}}
\label{appendix:lemma:lin_prec_formula}
\begin{lemma*}[Polytope assertion]
Given a finite set $\polytope \subset \ensembles$ of ensembles, let $S = \bigcup_{\ensemble \in \polytope} \supp(\ensemble)$, and for an ensemble $\ensemble \in \polytope$ write $\ket{\ensemble} \in \Reals^{|S|}$ to denote the ensemble (column) vector over the real vector space $\Reals^{|S|}$ with an orthonormal basis indexed by the hybrid states $\hstate \in S$,  i.e.,  $\braket{\hstate\mid\ensemble}= \ensemble(\hstate)$ for all $\hstate \in S$.
If the set of vectors $\mvals{\polytope}=\{\ket{\ensemble} \mid \ensemble \in \polytope\}$ is a linearly independent set of vectors, define the polytope assertion for  $\polytope$ by 
$$\pre_{\polytope} = \left(\probOf{\stateA_S} = 1
\;\land\; \bigwedge_{\ensemble \in \polytope} \pterm_\ensemble \geq 0
\;\land\; \bigwedge_{\hstate \in S} \Big(\probOf{\stateA_\hstate} = \sum_{\ensemble\in\polytope} \ensemble(\hstate)\cdot\pterm_\ensemble\Big)\right),$$
where $\pterm_\ensemble = \sum_{\hstate\in S} L_{\ensemble,\hstate}\cdot\probOf{\stateA_\hstate}$ is a probabilistic term using  the left pseudoinverse $L$  of the matrix $V \in \Reals^{|S|\times|\polytope|}$ with elements $V_{\hstate, \ensemble} = \ensemble(\hstate)$.
For all ensembles $\ensemble \in \ensembles$, we have $\ensemble \in \operatorname{conv}(\polytope)$ iff $\ensemble \vDash \pre_{\polytope}$.
\end{lemma*}
\begin{proof}
We note that since the vectors in $\mvals{\polytope}$ are linearly independent, $V$ has full column rank,  so the left pseudo inverse $L$ is guaranteed to exist~\cite{penrose55}.

\begin{itemize}
    \item ($\Rightarrow$) Let $\ensemble = \sum_{0 \leq i < |\polytope|} \prob_i\cdot \ensemble_i$ be a convex combination with real coefficients $\prob_i\geq 0$ and $\sum_i\prob_i = 1$, and 
    denote by $\ket{\prob} = [\prob_0, \ldots, \prob_{|\polytope|-1}]^T$ the column vector of these coefficients.
    Furthermore, we have $\ensemble \vDash \pre_\polytope$ iff
(a)~$\supp(\ensemble)\subseteq S$,
(b)~$\ket{\prob} = L\ket{\ensemble}$ is a vector with nonnegative elements, and
(c)~$\ket{\ensemble} = V\ket{\prob}$.
    For every ensemble $\ensemble$, it holds  $\sem{\probOf{\stateA_\hstate}}(\ensemble) = \ensemble(\hstate)$ and $\sem{\probOf{\stateA_S}}(\ensemble) = \sum_{\hstate\in S}\ensemble(\hstate)$.
Then (a) holds, and $\ket{\ensemble} = V\ket{\prob}$.
Therefore, $\ket{\prob}= L\cdot  V\ket{\prob} = L\ket{\ensemble}$. 
So $\ensemble$ also satisfies (b) and (c).
\item ($\Leftarrow$) By (a) and (c), we have $\ensemble = \sum_{0 \leq i < |\polytope|} \prob_i\cdot \ensemble_i$, and $\prob_i\geq 0$ by (b).
Each $\ensemble_i$ has total mass $1$ on $S$, so the columns of $V$ sum to $1$, and
$1 = \sum_{\hstate\in S}\ensemble(\hstate) = \sum_{\hstate} \bra{\hstate}(V\ket{\prob})= \sum_{i} \prob_i$.
Hence $\ensemble\in\operatorname{conv}(\polytope)$.
\end{itemize}
\end{proof}

\section{Appendix: Algorithm \textsc{checkLin}}
\label{appendix:checklin_alg}
\begin{algorithm}[H]
\caption{\textsc{checkLin}$(\MC,\prog,\pre,\post)$}
\label{alg:checkLin}
\scriptsize
\begin{algorithmic}[1]
\Require Ensemble graph $\MC$, loop-free program $\prog$, linear precondition $\pre$, postcondition $\post$

\ForAll{$\conv(\polytope) \in \prepolys(\pre)$} \Comment{Polytope $\conv(\polytope)$ has $N$ corners $\{\ensemble_0, \ensemble_1, \ldots, \ensemble_{N-1}\}$.}
	\For{ $i = 0$ to $N-1$}
		\State $\ensemble_i' \gets \getFinalEnsemble(\MC, \langle\prog,\ensemble_i\rangle)$
	\EndFor
    \If{$\forall (u_0,\ldots,u_{N-1}). ((\bigwedge_{i=0}^{N-1} u_i \ge 0 \land\sum_{i=0}^{N-1} u_i = 1)\Rightarrow(\sum_{i=0}^{N-1} u_i \cdot \ensemble_i' \vDash \post))$ is $\texttt{unsat}$}
        \State \Return \texttt{false}
    \EndIf
\EndFor
\State \Return \texttt{true}
\end{algorithmic}
\end{algorithm}

\section{Proof of Theorem~\ref{thm:linear-prec-verification}}
\label{proof:linear-prec-verification}

\begin{theorem*}[Correctness $\textsc{checkLin}$] For all hardware specifications $\hardware$, loop-free quantum programs $\prog$, and ensemble assertions $\pre$ and $\post$,
if $\pre$ is linear,
then $\htriple{\pre}{\prog}{\post}{\hardware}$ is valid iff $\textsc{checkLin}(EG(\hardware, \atomprogs), \prog, \pre, \post)=\texttt{true}$.
\end{theorem*}

\begin{proof}
Suppose we are given a hardware specification $\hardware$, the ensemble graph $\MC= EG(\hardware, \atomprogs)$, a loop-free quantum program $\prog$, a linear assertion $\pre$ with a set $\QSetPolytopes = \prepolys(\pre)$ of polytopes,  and a postcondition $\post$.
\begin{itemize}
\item [$(\Rightarrow)$] Suppose that $\htriple{\pre}{\prog}{\post}{\hardware}$ is valid. 
Then for every  $\conv(\polytope) = \conv(\{\ensemble_0, \ldots, \ensemble_{|\polytope|-1}\}) \in \QSetPolytopes$ and real coefficients $u_i\geq 0$ with $\sum_{i=0}^{|\polytope|-1}u_i =1$, it holds 
$\sem{\prog}_{\hardware}\left(\sum_{i=0}^{|\polytope|-1}u_i\cdot \ensemble_i\right) = \sum_{i=0}^{|\polytope|-1}u_i\cdot \sem{\prog}_{\hardware}(\ensemble_i)\vDash\post$ by linearity of program semantics. So  Algorithm~\ref{alg:checkLin} returns $\texttt{true}$.
	\item [$(\Leftarrow)$]  Suppose $\textsc{checkLin}(\MC, \prog, \pre, \post)=\texttt{true}$. 
	Then, linearity on program semantics implies that for all polytopes $\conv(\polytope) \in \QSetPolytopes$, and ensembles in $\ensemble\in \operatorname{conv}(\polytope)$, it holds $\sem{\prog}_\hardware(\ensemble) \vDash \post$ because Line~6 of Algorithm~\ref{alg:checkLin} is always satisfiable.  Follows $\htriple{\pre}{\prog}{\post}{\hardware}$ is valid.
\end{itemize}
\end{proof}

\section{Programs}
\label{appendix:programs}
% \subsubsection{Synthesis Experiment:\Paritybitflip~[IPMA]}
\begin{figure}[H]
        \centering
        \begin{lstlisting}[basicstyle=\ttfamily\scriptsize,keepspaces=true,showstringspaces=false,tabsize=4, mathescape=true]
$\CX$([q0,q2]); $\CX$([q1,q2]); 
$\cvar$ = $\Measure$(q2); 
if $\cvar$ = 0:
    $\halt$
if $\cvar$ = 1:
	$\cvar$ = $\Measure$(q2); 
	if $\cvar$ = 0:
		$\halt$
	if $\cvar$ = 1:
		$\X$([q0]) 
		$\halt$
        \end{lstlisting}
		\caption{Baseline program for the \Paritybitflip~problem with the $\IPMAIns$ instruction set for  horizon 4, using a majority-vote strategy. For larger horizons, the baseline adds further measurement rounds and applies the correction according to the majority of outcomes.}
\label{fig:alg_perfect_ipma2}
\end{figure}

\begin{figure}[H]
\begin{minipage}[t]{0.48\linewidth}
\begin{lstlisting}[basicstyle=\ttfamily\tiny,keepspaces=true,showstringspaces=false,tabsize=4, mathescape=true]
$\CX$([q0,q2]); $\CX$([q1,q2]); 
$\cvar$ = $\Measure$(q2); 
if $\cvar$ = 0:
	$\cvar$ = $\Measure$(q2); 
	if $\cvar$ = 0:
		$\cvar$ = $\Measure$(q2); 
		if $\cvar$ = 0:
			$\X$([q0])
            $\halt$
		if $\cvar$ = 1:
			$\cvar$ = $\Measure$(q2); 
			if $\cvar$ = 0:
				$\X$([q0])
				$\halt$
			if $\cvar$ = 1:
				$\cvar$ = $\Measure$(q2); 
				if $\cvar$ = 0:
					$\X$([q0]) 
					$\halt$
				if $\cvar$ = 1:
                    $\halt$
	if $\cvar$ = 1:
		$\cvar$ = $\Measure$(q2); 
		if $\cvar$ = 0:
			$\cvar$ = $\Measure$(q2); 
			if $\cvar$ = 0:
				$\X$([q0])
				$\halt$
			if $\cvar$ = 1:
				$\cvar$ = $\Measure$(q2); 
				if $\cvar$ = 0:
					$\X$([q0]) 
					$\halt$
				if $\cvar$ = 1:
					$\halt$
		if $\cvar$ = 1:
			$\cvar$ = $\Measure$(q2); 
			if $\cvar$ = 0:
				$\cvar$ = $\Measure$(q2); 
				if $\cvar$ = 0:
					$\X$([q0]) 
					$\halt$
				if $\cvar$ = 1:
                    $\halt$
			if $\cvar$ = 1:
               $\halt$
\end{lstlisting}
\end{minipage}\hfill
\begin{minipage}[t]{0.48\linewidth}
\begin{lstlisting}[basicstyle=\ttfamily\tiny,keepspaces=true,showstringspaces=false,tabsize=4, mathescape=true]
if $\cvar$ = 1:
	$\cvar$ = $\Measure$(q2); 
	if $\cvar$ = 0:
		$\cvar$ = $\Measure$(q2); 
		if $\cvar$ = 0:
			$\cvar$ = $\Measure$(q2); 
			if $\cvar$ = 0:
				$\X$([q0])
				$\halt$
			if $\cvar$ = 1:
				$\cvar$ = $\Measure$(q2); 
				if $\cvar$ = 0:
					$\X$([q0]) 
					$\halt$
				if $\cvar$ = 1:
					$\halt$
		if $\cvar$ = 1:
			$\cvar$ = $\Measure$(q2); 
			if $\cvar$ = 0:
				$\cvar$ = $\Measure$(q2); 
				if $\cvar$ = 0:
					$\X$([q0]) 
					$\halt$
				if $\cvar$ = 1:
					$\halt$
			if $\cvar$ = 1:
				$\halt$
	if $\cvar$ = 1:
		$\cvar$ = $\Measure$(q2); 
		if $\cvar$ = 0:
			$\cvar$ = $\Measure$(q2); 
			if $\cvar$ = 0:
				$\cvar$ = $\Measure$(q2); 
				if $\cvar$ = 0:
					$\X$([q0]) 
					$\halt$
				if $\cvar$ = 1:
					$\halt$
			if $\cvar$ = 1:
				$\halt$
		if $\cvar$ = 1:
			$\halt$
\end{lstlisting}
\end{minipage}
\caption{Program achieving the maximum improvement over the baseline (Figure~\ref{fig:alg_perfect_ipma2}) for the \Paritybitflip~problem with the $\IPMAIns$ instruction set and the singular precondition.}
\label{fig:alg_optimal_ipma2}
\end{figure}

\begin{figure}[H]
\tiny
\begin{minipage}[t]{0.48\linewidth}
\begin{lstlisting}[basicstyle=\ttfamily\scriptsize,keepspaces=true,showstringspaces=false,tabsize=4, mathescape=true]
$\CX$([q0,q2]); $\CX$([q1,q2]); 
$\cvar$ = $\Measure$(q2); 
if $\cvar$ = 0:
	$\cvar$ = $\Measure$(q2); 
	if $\cvar$ = 0:
		$\CX$([q0,q2]); $\CX$([q1,q2]); 
		$\cvar$ = $\Measure$(q2); 
		if $\cvar$ = 0:
			$\cvar$ = $\Measure$(q2); 
			if $\cvar$ = 0:
				$\cvar$ = $\Measure$(q2); 
				if $\cvar$ = 0:
					$\halt$
				if $\cvar$ = 1:
					$\halt$
			if $\cvar$ = 1:
				$\cvar$ = $\Measure$(q2); 
				if $\cvar$ = 0:
                    $\halt$
				if $\cvar$ = 1:
					$\X$([q0]) 
					$\halt$
		if $\cvar$ = 1:
			$\CX$([q0,q2]); $\CX$([q1,q2]); 
			$\cvar$ = $\Measure$(q2); 
			if $\cvar$ = 0:
				$\X$([q0]) 
				$\halt$
			if $\cvar$ = 1:
                $\halt$
	if $\cvar$ = 1:
		$\cvar$ = $\Measure$(q2); 
		if $\cvar$ = 0:
			$\cvar$ = $\Measure$(q2); 
			if $\cvar$ = 0:
				$\cvar$ = $\Measure$(q2); 
				if $\cvar$ = 0:
                    $\halt$
				if $\cvar$ = 1:
					$\cvar$ = $\Measure$(q2); 
					if $\cvar$ = 0:
						$\halt$
					if $\cvar$ = 1:
						$\X$([q0]) 
						$\halt$
			if $\cvar$ = 1:
				$\cvar$ = $\Measure$(q2); 
				if $\cvar$ = 0:
					$\cvar$ = $\Measure$(q2); 
					if $\cvar$ = 0:
						$\halt$
					if $\cvar$ = 1:
						$\X$([q0]) 
						$\halt$
				if $\cvar$ = 1:
					$\X$([q0])
                    $\halt$
		if $\cvar$ = 1:
			$\cvar$ = $\Measure$(q2); 
			if $\cvar$ = 0:
				$\cvar$ = $\Measure$(q2); 
				if $\cvar$ = 0:
					$\cvar$ = $\Measure$(q2); 
					if $\cvar$ = 0:
						$\halt$
					if $\cvar$ = 1:
						$\X$([q0]) 
						$\halt$
				if $\cvar$ = 1:
					$\X$([q0])
					$\halt$
			if $\cvar$ = 1:
				$\X$([q0])
                $\halt$
\end{lstlisting}
\end{minipage}\hfill
\begin{minipage}[t]{0.48\linewidth}
\begin{lstlisting}[basicstyle=\ttfamily\scriptsize,keepspaces=true,showstringspaces=false,tabsize=4, mathescape=true]
if $\cvar$ = 1:
	$\cvar$ = $\Measure$(q2); 
	if $\cvar$ = 0:
		$\cvar$ = $\Measure$(q2); 
		if $\cvar$ = 0:
			$\cvar$ = $\Measure$(q2); 
			if $\cvar$ = 0:
				$\cvar$ = $\Measure$(q2); 
				if $\cvar$ = 0:
				    $\halt$
				if $\cvar$ = 1:
					$\cvar$ = $\Measure$(q2); 
					if $\cvar$ = 0:
						$\halt$
					if $\cvar$ = 1:
						$\X$([q0]) 
						$\halt$
			if $\cvar$ = 1:
				$\cvar$ = $\Measure$(q2); 
				if $\cvar$ = 0:
					$\cvar$ = $\Measure$(q2); 
					if $\cvar$ = 0:
						$\halt$
					if $\cvar$ = 1:
						$\X$([q0]) 
						$\halt$
				if $\cvar$ = 1:
					$\X$([q0])
                    $\halt$
		if $\cvar$ = 1:
			$\cvar$ = $\Measure$(q2); 
			if $\cvar$ = 0:
				$\cvar$ = $\Measure$(q2); 
				if $\cvar$ = 0:
					$\cvar$ = $\Measure$(q2); 
					if $\cvar$ = 0:
						$\halt$
					if $\cvar$ = 1:
						$\X$([q0]) 
						$\halt$
				if $\cvar$ = 1:
					$\X$([q0])
					$\halt$
			if $\cvar$ = 1:
				$\X$([q0])
                $\halt$
	if $\cvar$ = 1:
		$\CX$([q0,q2]); $\CX$([q1,q2]); 
		$\cvar$ = $\Measure$(q2); 
		if $\cvar$ = 0:
			$\cvar$ = $\Measure$(q2); 
			if $\cvar$ = 0:
				$\X$([q0])
				$\halt$
			if $\cvar$ = 1:
				$\cvar$ = $\Measure$(q2); 
				if $\cvar$ = 0:
					$\X$([q0]) 
					$\halt$
				if $\cvar$ = 1:
                    $\halt$
		if $\cvar$ = 1:
			$\CX$([q0,q2]); $\CX$([q1,q2]); 
			$\cvar$ = $\Measure$(q2); 
			if $\cvar$ = 0:
				$\X$([q0]) 
				$\halt$
			if $\cvar$ = 1:
				$\halt$
\end{lstlisting}
\end{minipage}
\caption{
Program achieving a substantial improvement in a robust embedding over the baseline (Figure~\ref{fig:alg_perfect_ipma2}) for the \Paritybitflip~problem with the $\IPMAIns$ instruction set and the singular precondition.}
\label{fig:alg_optimal_ipma2_cambridge}
\end{figure}

\begin{figure}[H]
\centering
\begin{minipage}[b]{0.45\linewidth}
\begin{lstlisting}[basicstyle=\ttfamily\scriptsize,
                   keepspaces=true,
                   showstringspaces=false,
                   tabsize=4,
                   mathescape=true,
                   breaklines=true]
$\CX$([q0,q2]); $\CX$([q1,q2]); 
$\cvar$ = $\Measure$(q2); 
if $\cvar$ = 0:
	$\cvar$ = $\Measure$(q2); 
	if $\cvar$ = 0:
		$\halt$
	if $\cvar$ = 1:
	    $\CX$([q0,q2]); $\CX$([q1,q2]); 
		$\halt$
if $\cvar$ = 1:
	$\X$([q0]) 
	$\halt$
\end{lstlisting}
\end{minipage}
\caption{Program achieving both the maximum improvement and a substantial improvement in a robust embedding over the baseline (Figure~\ref{fig:alg_perfect_ipma2}) for the \Paritybitflip~problem with the $\IPMAIns$ instruction set and the linear  precondition.}
\label{fig:ipma_lin_programs}
\end{figure}

% \subsubsection{Synthesis Experiment: \Paritybitflip~[CX+H]}
\begin{figure}[H]
\centering
\begin{minipage}[b]{0.30\linewidth}
\vbox to 4cm{%
\begin{lstlisting}[basicstyle=\ttfamily\scriptsize,
                   keepspaces=true,
                   showstringspaces=false,
                   tabsize=4,
                   mathescape=true,
                   breaklines=true]
$\Had$([q2]); 
$\CX$([q2,q1]); 
$\Had$([q2]); 
$\Had$([q1]); 
$\CX$([q2,q1]); 
$\CX$([q0,q1]); 
$\cvar$ = $\Measure$(q2);
$\halt$
\end{lstlisting}
\vfill
}
\end{minipage}\hfill
\begin{minipage}[b]{0.30\linewidth}
\vbox to 4cm{%
\begin{lstlisting}[basicstyle=\ttfamily\scriptsize,
                   keepspaces=true,
                   showstringspaces=false,
                   tabsize=4,
                   mathescape=true,
                   breaklines=true]
$\Had$([q2]); 
$\CX$([q2,q1]); 
$\Had$([q2]); 
$\CX$([q0,q1]); 
$\Had$([q1]); 
$\CX$([q0,q1]); 
$\CX$([q2,q1]); 
$\halt$
\end{lstlisting}
\vfill
}
\end{minipage}\hfill
\begin{minipage}[b]{0.30\linewidth}
\vbox to 4cm{%
\begin{lstlisting}[basicstyle=\ttfamily\scriptsize,
                   keepspaces=true,
                   showstringspaces=false,
                   tabsize=4,
                   mathescape=true,
                   breaklines=true]
$\Had$([q2]); 
$\CX$([q2,q1]); 
$\Had$([q2]); 
$\Had$([q1]); 
$\CX$([q0,q1]); 
$\cvar$ = $\Measure$(q2); 
if $\cvar$ = 0:
    $\Had$([q2]); 
    $\halt$
if $\cvar$ = 1:
    $\CX$([q2,q1]); 
    $\halt$
\end{lstlisting}
\vfill
}
\end{minipage}
\caption{Synthesized program using the $\CXHIns$ instruction set for the \Paritybitflip~problem. 
\textbf{(a)} Program synthesized for the perfect hardware. 
\textbf{(b)} Optimal program for a noisy embedding with the maximum improvement with respect to (a). 
\textbf{(c)} Program achieving a substantial improvement in a robust embedding with respect to (a).}
\label{fig:cxh_programs}
\end{figure}

% \subsubsection{Synthesis Experiment: {\sc GHZ state-preparation}}
\begin{figure}[H]
	\includegraphics[width=0.75\linewidth]{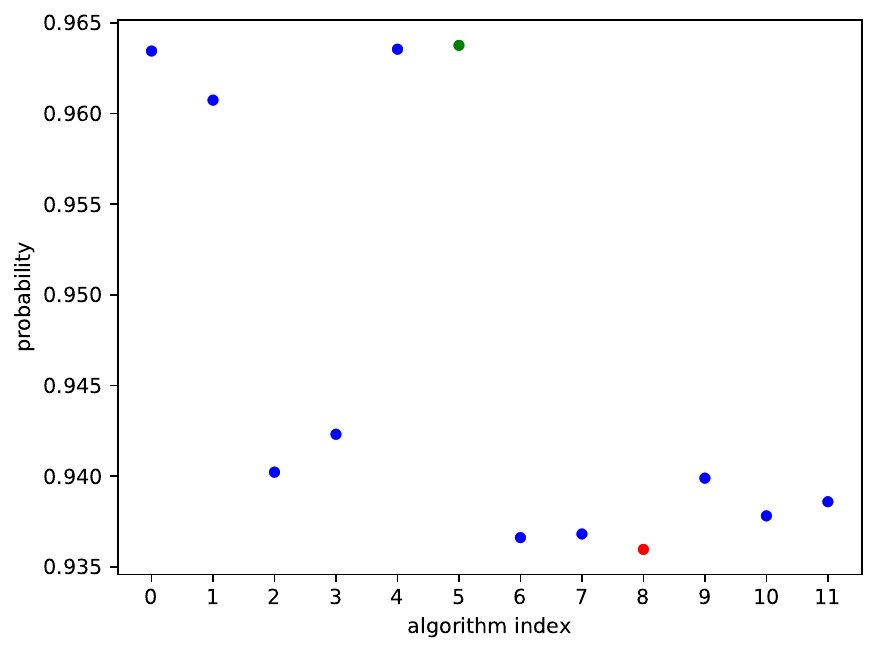}
	\Description{Accuracies achieved by all 12 different programs synthesized for the \GHZsp~problem.
	It considers a specific embedding of Yorktown quantum hardware.
	The point in green represents the optimally-synthesized program,
	whereas the point in red is the one synthesized for the perfect hardware.
	}
	\caption{Accuracies achieved by all 12 different programs synthesized for the \GHZsp~problem on an  embedding of Yorktown quantum hardware.
	The point in green represents the optimally-synthesized program,
	whereas the point in red is the one synthesized for the noiseless hardware.
	}
	\label{fig:ghz_experiment}
\end{figure}

% \subsubsection{Synthesis Experiment: {\sc State-reset}}
\begin{figure}[H]
    \centering
    \begin{subfigure}[b]{0.4\textwidth}
        \centering
        \begin{lstlisting}[basicstyle=\ttfamily\tiny,keepspaces=true,showstringspaces=false,tabsize=4, mathescape=true]
$\cvar$ = $\Measure$(q0) 
if $\cvar$ = 0:
    $\halt$
if $\cvar$ = 1:
	$\cvar$ = $\Measure$(q0) 
	if $\cvar$ = 0:
		$\halt$
	if $\cvar$ = 1:
		$\X$([q0]) 
		$\halt$
		\end{lstlisting}
        \label{fig:alg-reset-baseline}
    \end{subfigure}
    \hfill
    \begin{subfigure}[b]{0.4\textwidth}
        \centering
        \begin{lstlisting}[basicstyle=\ttfamily\tiny,keepspaces=true,showstringspaces=false,tabsize=4, mathescape=true]
$\X$([q0])
$\cvar$ = $\Measure$(q0) 
if $\cvar$ = 0:
    $\halt$
if $\cvar$ = 1:
    $\X$([q0]) 
    $\halt$
		\end{lstlisting}
        \label{fig:alg-reset-robust}
    \end{subfigure}
    \caption{\textbf{(a)}~Baseline program for the \StateReset~problem for  horizon 4, using a majority-vote strategy. For larger horizons, the baseline adds further measurement rounds and applies the correction according to the majority of outcomes. 
    \textbf{(b)}~Program achieving a substantial improvement in a robust embedding over (a) for the \StateReset~problem with linear precondition.}
    \label{fig:reset_algs}
\end{figure}

\begin{figure}[H]
\tiny
\begin{minipage}[t]{0.48\linewidth}
\begin{lstlisting}[basicstyle=\ttfamily\scriptsize,keepspaces=true,showstringspaces=false,tabsize=4, mathescape=true]
$\cvar$ = $\Measure$(q0) 
if $\cvar$ = 0:
	$\X$([q0]) 
	$\cvar$ = $\Measure$(q0) 
	if $\cvar$ = 0:
		$\cvar$ = $\Measure$(q0) 
		if $\cvar$ = 0:
			$\cvar$ = $\Measure$(q0) 
			if $\cvar$ = 0:
				$\X$([q0]) 
				$\halt$
			if $\cvar$ = 1:
				$\cvar$ = $\Measure$(q0) 
				if $\cvar$ = 0:
					$\X$([q0]) 
					$\halt$
				if $\cvar$ = 1:
					$\halt$
		if $\cvar$ = 1:
			$\cvar$ = $\Measure$(q0) 
			if $\cvar$ = 0:
				$\cvar$ = $\Measure$(q0) 
				if $\cvar$ = 0:
					$\X$([q0]) 
					$\halt$
				if $\cvar$ = 1:
					$\halt$
			if $\cvar$ = 1:
				$\halt$
	if $\cvar$ = 1:
		$\cvar$ = $\Measure$(q0) 
		if $\cvar$ = 0:
			$\cvar$ = $\Measure$(q0) 
			if $\cvar$ = 0:
				$\cvar$ = $\Measure$(q0) 
				if $\cvar$ = 0:
					$\X$([q0]) 
					$\halt$
				if $\cvar$ = 1:
					$\halt$
			if $\cvar$ = 1:
				$\halt$
		if $\cvar$ = 1:
			$\halt$
\end{lstlisting}
\end{minipage}\hfill
\begin{minipage}[t]{0.48\linewidth}
\begin{lstlisting}[basicstyle=\ttfamily\scriptsize,keepspaces=true,showstringspaces=false,tabsize=4, mathescape=true]
if $\cvar$ = 1:
	$\cvar$ = $\Measure$(q0) 
	if $\cvar$ = 0:
		$\cvar$ = $\Measure$(q0) 
		if $\cvar$ = 0:
			$\X$([q0]) 
			$\cvar$ = $\Measure$(q0) 
			if $\cvar$ = 0:
				$\cvar$ = $\Measure$(q0) 
				if $\cvar$ = 0:
					$\X$([q0]) 
					$\halt$
				if $\cvar$ = 1:
					$\halt$
			if $\cvar$ = 1:
				$\halt$
		if $\cvar$ = 1:
			$\cvar$ = $\Measure$(q0) 
			if $\cvar$ = 0:
				$\X$([q0]) 
				$\cvar$ = $\Measure$(q0) 
				if $\cvar$ = 0:
					$\X$([q0]) 
					$\halt$
				if $\cvar$ = 1:
					$\halt$
			if $\cvar$ = 1:
				$\halt$
	if $\cvar$ = 1:
		$\cvar$ = $\Measure$(q0) 
		if $\cvar$ = 0:
			$\cvar$ = $\Measure$(q0) 
			if $\cvar$ = 0:
				$\X$([q0]) 
				$\cvar$ = $\Measure$(q0) 
				if $\cvar$ = 0:
					$\X$([q0]) 
					$\halt$
				if $\cvar$ = 1:
					$\halt$
			if $\cvar$ = 1:
				$\halt$
		if $\cvar$ = 1:
			$\halt$
\end{lstlisting}
\end{minipage}
\caption{
Program achieving the maximum improvement over the baseline (Figure~\ref{fig:reset_algs}(a)) for the \StateReset~problem with linear precondition.}
\label{fig:alg_optimal_reset}
\end{figure}

\begin{figure}[H]
    \centering
    \begin{lstlisting}[basicstyle=\ttfamily\tiny,keepspaces=true,showstringspaces=false,tabsize=4, mathescape=true]
$\cvar$ = $\Measure$(q1); 
if $\cvar$ = 0:
    $\Z$([q0])
if $\cvar$ = 1:
	$\Z$([q0])
    $\X$([q0]) 
    $\halt$
		\end{lstlisting}
    \caption{
    Baseline program for the \PhasePrep~problem for  horizon 3, using a majority-vote strategy. For larger horizons, the baseline adds further measurement rounds and applies the correction according to the majority of outcomes.}
    \label{fig:lphase-baseline}
\end{figure}

\begin{figure}[H]
\tiny
\begin{minipage}[t]{0.48\linewidth}
\begin{lstlisting}[basicstyle=\ttfamily\scriptsize,keepspaces=true,showstringspaces=false,tabsize=4, mathescape=true]
$\cvar$ = $\Measure$(q1)
if $\cvar$ = 0:
	$\cvar$ = $\Measure$(q1)
	if $\cvar$ = 0:
		$\cvar$ = $\Measure$(q1)
		if $\cvar$ = 0:
			$\cvar$ = $\Measure$(q1)
			if $\cvar$ = 0:
				$\Z$([q0]) 
				$\X$([q0]) 
				$\halt$
			if $\cvar$ = 1:
				$\Z$([q0]) 
				$\X$([q0]) 
				$\halt$
		if $\cvar$ = 1:
			$\cvar$ = $\Measure$(q1)
			if $\cvar$ = 0:
				$\Z$([q0]) 
				$\X$([q0]) 
				$\halt$
			if $\cvar$ = 1:
				$\cvar$ = $\Measure$(q1)
				if $\cvar$ = 0:
					$\Z$([q0]) 
					$\X$([q0]) 
					$\halt$
				if $\cvar$ = 1:
					$\Z$([q0]) 
					$\halt$
	if $\cvar$ = 1:
		$\cvar$ = $\Measure$(q1)
		if $\cvar$ = 0:
			$\cvar$ = $\Measure$(q1)
			if $\cvar$ = 0:
				$\Z$([q0]) 
				$\X$([q0]) 
				$\halt$
			if $\cvar$ = 1:
				$\cvar$ = $\Measure$(q1)
				if $\cvar$ = 0:
					$\Z$([q0]) 
					$\X$([q0]) 
					$\halt$
				if $\cvar$ = 1:
					$\Z$([q0]) 
					$\halt$
		if $\cvar$ = 1:
			$\cvar$ = $\Measure$(q1)
			if $\cvar$ = 0:
				$\cvar$ = $\Measure$(q1)
				if $\cvar$ = 0:
					$\Z$([q0]) 
					$\X$([q0]) 
					$\halt$
				if $\cvar$ = 1:
					$\Z$([q0]) 
					$\halt$
			if $\cvar$ = 1:
				$\cvar$ = $\Measure$(q1)
				if $\cvar$ = 0:
					$\Z$([q0]) 
					$\halt$
				if $\cvar$ = 1:
					$\Z$([q0]) 
					$\halt$
\end{lstlisting}
\end{minipage}\hfill
\begin{minipage}[t]{0.48\linewidth}
\begin{lstlisting}[basicstyle=\ttfamily\scriptsize,keepspaces=true,showstringspaces=false,tabsize=4, mathescape=true]
if $\cvar$ = 1:
	$\cvar$ = $\Measure$(q1)
	if $\cvar$ = 0:
		$\cvar$ = $\Measure$(q1)
		if $\cvar$ = 0:
			$\cvar$ = $\Measure$(q1)
			if $\cvar$ = 0:
				$\Z$([q0]) 
				$\X$([q0]) 
				$\halt$
			if $\cvar$ = 1:
				$\cvar$ = $\Measure$(q1)
				if $\cvar$ = 0:
					$\Z$([q0]) 
					$\X$([q0]) 
					$\halt$
				if $\cvar$ = 1:
					$\Z$([q0]) 
					$\halt$
		if $\cvar$ = 1:
			$\cvar$ = $\Measure$(q1)
			if $\cvar$ = 0:
				$\cvar$ = $\Measure$(q1)
				if $\cvar$ = 0:
					$\Z$([q0]) 
					$\X$([q0]) 
					$\halt$
				if $\cvar$ = 1:
					$\Z$([q0]) 
					$\halt$
			if $\cvar$ = 1:
				$\cvar$ = $\Measure$(q1)
				if $\cvar$ = 0:
					$\Z$([q0]) 
					$\halt$
				if $\cvar$ = 1:
					$\Z$([q0]) 
					$\halt$
	if $\cvar$ = 1:
		$\cvar$ = $\Measure$(q1)
		if $\cvar$ = 0:
			$\cvar$ = $\Measure$(q1)
			if $\cvar$ = 0:
				$\cvar$ = $\Measure$(q1)
				if $\cvar$ = 0:
					$\Z$([q0]) 
					$\X$([q0]) 
					$\halt$
				if $\cvar$ = 1:
					$\Z$([q0]) 
					$\halt$
			if $\cvar$ = 1:
				$\cvar$ = $\Measure$(q1)
				if $\cvar$ = 0:
					$\Z$([q0]) 
					$\halt$
				if $\cvar$ = 1:
					$\Z$([q0]) 
					$\halt$
		if $\cvar$ = 1:
			$\Z$([q0]) 
			$\Z$([q0]) 
			$\Z$([q0]) 
			$\halt$
\end{lstlisting}
\end{minipage}
\caption{
Program achieving the maximum improvement over the baseline (Figure~\ref{fig:lphase-baseline}) for the \PhasePrep~problem with singular precondition.}
\label{fig:alg_optimal_lphase}
\end{figure}

\begin{figure}[H]
\tiny
\begin{minipage}[t]{0.48\linewidth}
\begin{lstlisting}[basicstyle=\ttfamily\scriptsize,keepspaces=true,showstringspaces=false,tabsize=4, mathescape=true]
$\cvar$ = $\Measure$(q1) 
if $\cvar$ = 0:
	$\cvar$ = $\Measure$(q1) 
	if $\cvar$ = 0:
		$\cvar$ = $\Measure$(q1) 
		if $\cvar$ = 0:
			$\Z$([q0]) 
			$\Z$([q0]) 
			$\Z$([q0]) 
			$\halt$
		if $\cvar$ = 1:
			$\cvar$ = $\Measure$(q1) 
			if $\cvar$ = 0:
				$\cvar$ = $\Measure$(q1) 
				if $\cvar$ = 0:
					$\Z$([q0]) 
					$\halt$
				if $\cvar$ = 1:
					$\Z$([q0]) 
					$\halt$
			if $\cvar$ = 1:
				$\cvar$ = $\Measure$(q1) 
				if $\cvar$ = 0:
					$\Z$([q0]) 
					$\halt$
				if $\cvar$ = 1:
					$\Z$([q0]) 
					$\X$([q0]) 
					$\halt$
	if $\cvar$ = 1:
		$\cvar$ = $\Measure$(q1) 
		if $\cvar$ = 0:
			$\cvar$ = $\Measure$(q1) 
			if $\cvar$ = 0:
				$\cvar$ = $\Measure$(q1) 
				if $\cvar$ = 0:
					$\Z$([q0]) 
					$\halt$
				if $\cvar$ = 1:
					$\Z$([q0]) 
					$\halt$
			if $\cvar$ = 1:
				$\cvar$ = $\Measure$(q1) 
				if $\cvar$ = 0:
					$\Z$([q0]) 
					$\halt$
				if $\cvar$ = 1:
					$\Z$([q0]) 
					$\X$([q0]) 
					$\halt$
		if $\cvar$ = 1:
			$\cvar$ = $\Measure$(q1) 
			if $\cvar$ = 0:
				$\cvar$ = $\Measure$(q1) 
				if $\cvar$ = 0:
					$\Z$([q0]) 
					$\halt$
				if $\cvar$ = 1:
					$\Z$([q0]) 
					$\X$([q0]) 
					$\halt$
			if $\cvar$ = 1:
				$\Z$([q0]) 
				$\X$([q0]) 
				$\halt$
\end{lstlisting}
\end{minipage}\hfill
\begin{minipage}[t]{0.48\linewidth}
\begin{lstlisting}[basicstyle=\ttfamily\scriptsize,keepspaces=true,showstringspaces=false,tabsize=4, mathescape=true]
if $\cvar$ = 1:
	$\cvar$ = $\Measure$(q1) 
	if $\cvar$ = 0:
		$\cvar$ = $\Measure$(q1) 
		if $\cvar$ = 0:
			$\cvar$ = $\Measure$(q1) 
			if $\cvar$ = 0:
				$\cvar$ = $\Measure$(q1) 
				if $\cvar$ = 0:
					$\Z$([q0]) 
					$\halt$
				if $\cvar$ = 1:
					$\Z$([q0]) 
					$\halt$
			if $\cvar$ = 1:
				$\cvar$ = $\Measure$(q1) 
				if $\cvar$ = 0:
					$\Z$([q0]) 
					$\halt$
				if $\cvar$ = 1:
					$\Z$([q0]) 
					$\X$([q0]) 
					$\halt$
		if $\cvar$ = 1:
			$\cvar$ = $\Measure$(q1) 
			if $\cvar$ = 0:
				$\cvar$ = $\Measure$(q1) 
				if $\cvar$ = 0:
					$\Z$([q0]) 
					$\halt$
				if $\cvar$ = 1:
					$\Z$([q0]) 
					$\X$([q0]) 
					$\halt$
			if $\cvar$ = 1:
				$\Z$([q0]) 
				$\X$([q0]) 
				$\halt$
	if $\cvar$ = 1:
		$\cvar$ = $\Measure$(q1) 
		if $\cvar$ = 0:
			$\cvar$ = $\Measure$(q1) 
			if $\cvar$ = 0:
				$\cvar$ = $\Measure$(q1) 
				if $\cvar$ = 0:
					$\Z$([q0]) 
					$\halt$
				if $\cvar$ = 1:
					$\Z$([q0]) 
					$\X$([q0]) 
					$\halt$
			if $\cvar$ = 1:
				$\Z$([q0]) 
				$\X$([q0]) 
				$\halt$
		if $\cvar$ = 1:
			$\cvar$ = $\Measure$(q1) 
			if $\cvar$ = 0:
				$\Z$([q0]) 
				$\X$([q0]) 
				$\halt$
			if $\cvar$ = 1:
				$\Z$([q0]) 
				$\X$([q0]) 
				$\halt$
\end{lstlisting}
\end{minipage}
\caption{
Program achieving a substantial improvement on a robust embedding over the baseline (Figure~\ref{fig:lphase-baseline}) for the \PhasePrep~problem with singular precondition.}
\label{fig:alg_robust_lphase}
\end{figure}

%%%%%
\begin{figure}[H]
\tiny
\begin{minipage}[t]{0.48\linewidth}
\begin{lstlisting}[basicstyle=\ttfamily\scriptsize,keepspaces=true,showstringspaces=false,tabsize=4, mathescape=true]
$\cvar$ = $\Measure$(q1) 
if $\cvar$ = 0:
	$\cvar$ = $\Measure$(q1) 
	if $\cvar$ = 0:
		$\cvar$ = $\Measure$(q1) 
		if $\cvar$ = 0:
			$\Z$([q0]) 
			$\X$([q0]) 
			$\halt$
		if $\cvar$ = 1:
			$\cvar$ = $\Measure$(q1) 
			if $\cvar$ = 0:
				$\Z$([q0]) 
				$\X$([q0]) 
				$\halt$
			if $\cvar$ = 1:
				$\Z$([q0]) 
				$\halt$
	if $\cvar$ = 1:
		$\cvar$ = $\Measure$(q1) 
		if $\cvar$ = 0:
			$\Z$([q0]) 
			$\X$([q0]) 
			$\halt$
		if $\cvar$ = 1:
			$\Z$([q0]) 
			$\X$([q0]) 
			$\X$([q0]) 
			$\halt$
\end{lstlisting}
\end{minipage}\hfill
\begin{minipage}[t]{0.48\linewidth}
\begin{lstlisting}[basicstyle=\ttfamily\scriptsize,keepspaces=true,showstringspaces=false,tabsize=4, mathescape=true]
if $\cvar$ = 1:
	$\cvar$ = $\Measure$(q1) 
	if $\cvar$ = 0:
		$\cvar$ = $\Measure$(q1) 
		if $\cvar$ = 0:
			$\Z$([q0]) 
			$\X$([q0]) 
			$\halt$
		if $\cvar$ = 1:
			$\Z$([q0]) 
			$\X$([q0]) 
			$\X$([q0]) 
			$\halt$
	if $\cvar$ = 1:
		$\cvar$ = $\Measure$(q1) 
		if $\cvar$ = 0:
			$\cvar$ = $\Measure$(q1) 
			if $\cvar$ = 0:
				$\Z$([q0]) 
				$\X$([q0]) 
				$\halt$
			if $\cvar$ = 1:
				$\Z$([q0]) 
				$\halt$
		if $\cvar$ = 1:
			$\Z$([q0]) 
			$\halt$
\end{lstlisting}
\end{minipage}
\caption{Program achieving the maximum improvement over the baseline (Figure~\ref{fig:lphase-baseline}) for the \PhasePrep~problem with linear precondition.}
\label{fig:alg_optimal_lphase_lin}
\end{figure}

\begin{figure}[H]
\tiny
\begin{minipage}[t]{0.48\linewidth}
\begin{lstlisting}[basicstyle=\ttfamily\scriptsize,keepspaces=true,showstringspaces=false,tabsize=4, mathescape=true]
$\cvar$ = $\Measure$(q1)
if $\cvar$ = 0:
	$\cvar$ = $\Measure$(q1)
	if $\cvar$ = 0:
		$\cvar$ = $\Measure$(q1)
		if $\cvar$ = 0:
			$\cvar$ = $\Measure$(q1)
			if $\cvar$ = 0:
				$\Z$([q0]) 
				$\halt$
			if $\cvar$ = 1:
				$\Z$([q0]) 
				$\X$([q0]) 
				$\X$([q0]) 
				$\halt$
		if $\cvar$ = 1:
			$\cvar$ = $\Measure$(q1)
			if $\cvar$ = 0:
				$\Z$([q0]) 
				$\X$([q0]) 
				$\X$([q0]) 
				$\halt$
			if $\cvar$ = 1:
				$\cvar$ = $\Measure$(q1)
				if $\cvar$ = 0:
					$\Z$([q0]) 
					$\halt$
				if $\cvar$ = 1:
					$\Z$([q0]) 
					$\X$([q0]) 
					$\halt$
	if $\cvar$ = 1:
		$\cvar$ = $\Measure$(q1)
		if $\cvar$ = 0:
			$\cvar$ = $\Measure$(q1)
			if $\cvar$ = 0:
				$\Z$([q0]) 
				$\X$([q0]) 
				$\X$([q0]) 
				$\halt$
			if $\cvar$ = 1:
				$\cvar$ = $\Measure$(q1)
				if $\cvar$ = 0:
					$\Z$([q0]) 
					$\halt$
				if $\cvar$ = 1:
					$\Z$([q0]) 
					$\X$([q0]) 
					$\halt$
		if $\cvar$ = 1:
			$\cvar$ = $\Measure$(q1)
			if $\cvar$ = 0:
				$\cvar$ = $\Measure$(q1)
				if $\cvar$ = 0:
					$\Z$([q0]) 
					$\halt$
				if $\cvar$ = 1:
					$\Z$([q0]) 
					$\X$([q0]) 
					$\halt$
			if $\cvar$ = 1:
				$\Z$([q0]) 
				$\X$([q0]) 
				$\halt$
\end{lstlisting}
\end{minipage}\hfill
\begin{minipage}[t]{0.48\linewidth}
\begin{lstlisting}[basicstyle=\ttfamily\scriptsize,keepspaces=true,showstringspaces=false,tabsize=4, mathescape=true]
if $\cvar$ = 1:
	$\cvar$ = $\Measure$(q1)
	if $\cvar$ = 0:
		$\cvar$ = $\Measure$(q1)
		if $\cvar$ = 0:
			$\cvar$ = $\Measure$(q1)
			if $\cvar$ = 0:
				$\Z$([q0]) 
				$\X$([q0]) 
				$\X$([q0]) 
				$\halt$
			if $\cvar$ = 1:
				$\cvar$ = $\Measure$(q1)
				if $\cvar$ = 0:
					$\Z$([q0]) 
					$\halt$
				if $\cvar$ = 1:
					$\Z$([q0]) 
					$\X$([q0]) 
					$\halt$
		if $\cvar$ = 1:
			$\cvar$ = $\Measure$(q1)
			if $\cvar$ = 0:
				$\cvar$ = $\Measure$(q1)
				if $\cvar$ = 0:
					$\Z$([q0]) 
					$\halt$
				if $\cvar$ = 1:
					$\Z$([q0]) 
					$\X$([q0]) 
					$\halt$
			if $\cvar$ = 1:
				$\Z$([q0]) 
				$\X$([q0]) 
				$\halt$
	if $\cvar$ = 1:
		$\cvar$ = $\Measure$(q1)
		if $\cvar$ = 0:
			$\cvar$ = $\Measure$(q1)
			if $\cvar$ = 0:
				$\cvar$ = $\Measure$(q1)
				if $\cvar$ = 0:
					$\Z$([q0]) 
					$\halt$
				if $\cvar$ = 1:
					$\Z$([q0]) 
					$\X$([q0]) 
					$\halt$
			if $\cvar$ = 1:
				$\Z$([q0]) 
				$\X$([q0]) 
				$\halt$
		if $\cvar$ = 1:
			$\Z$([q0]) 
			$\Z$([q0]) 
			$\Z$([q0]) 
			$\X$([q0]) 
			$\halt$
\end{lstlisting}
\end{minipage}
\caption{Program achieving a substantial improvement on a robust embedding over the baseline (Figure~\ref{fig:lphase-baseline}) for the \PhasePrep~problem with linear precondition.}
\label{fig:alg_robust_lphase_lin}
\end{figure}

\end{document}